\documentclass[aps,twocolumn,pra,superscriptaddress]{revtex4-2}
\usepackage{graphicx}
\usepackage{amsthm}
\usepackage{amsmath}
\usepackage{amssymb}
\usepackage{times}
\usepackage{epstopdf}
\usepackage[english]{babel}
\usepackage{natbib}
\usepackage{hyperref}
\usepackage{physics}
\hypersetup{
    colorlinks=true,       
    linkcolor=cyan,          
    citecolor=magenta,        
    filecolor=magenta,      
    urlcolor=cyan,           
    runcolor=cyan
}
\usepackage{cleveref}
\usepackage{tablefootnote}
\usepackage{verbatim}
\usepackage{subfigure}
\usepackage{placeins}

\usepackage[pass]{geometry}

\usepackage[normalem]{ulem}
\usepackage{mathtools}
\newtheorem{theorem}{Theorem}
\newtheorem{lemma}{Lemma}

\newcommand{\A}[1]{\hat{a}_{#1}}
\newcommand{\Ad}[1]{\hat{a}^\dagger_{#1}}

\newcommand{\B}[1]{\hat{b}_{#1}}
\newcommand{\Bd}[1]{\hat{b}^\dagger_{#1}}

\newcommand{\D}[2][]{\mathcal{\hat{D}}_{#1}\pqty{#2}}
\newcommand{\Dd}[2][]{\mathcal{\hat{D}}^\dagger_{#1}\pqty{#2}}
\newcommand{\oo}{\mathrm{o}}
\newcommand{\out}{\mathrm{out}}
\newcommand{\iin}{\mathrm{in}}
\newcommand{\ee}{\mathrm{e}}
\newcommand{\mm}{\mathrm{m}}
\newcommand{\Ain}[1]{a_{#1,\mathrm{in}}}
\newcommand{\Adin}[1]{a^\dagger_{#1,\mathrm{in}}}

\newcommand{\w}{\omega}

\newcommand{\Ab}{\mathbf{A}}
\newcommand{\Bb}{\mathbf{B}}
\newcommand{\AD}{\mathbf{A_\mathrm{d}}}
\newcommand{\Am}{\mathbf{A_-}}
\newcommand{\Ap}{\mathbf{A_+}}
\newcommand{\Apm}{\mathbf{A_\pm}}
\newcommand{\Amp}{\mathbf{A_\mp}}
\newcommand{\Xm}{\mathbf{\Xi}_-}
\newcommand{\Xp}{\mathbf{\Xi}_+}
\newcommand{\XI}[2]{\mathbf{\Xi}^{\bqty{#1}}_{#2}}
\newcommand{\Xpm}{\mathbf{\Xi}_\pm}

\newcommand{\II}{\mathbf{I}}
\newcommand{\TT}{\mathbf{T}}
\newcommand{\XX}[2]{\mathbf{X}^{\bqty{#1}}_{#2}}
\newcommand{\diag}{\mathrm{diag}}

\begin{document}
\title{Parametric Amplification of an Optomechanical Quantum Interconnect}
\author{Huo Chen}
\email{huochen@lbl.gov}
\affiliation{Applied Mathematics and Computational Research Division, Lawrence Berkeley National Laboratory, Berkeley, CA 94720, USA}

\author{Marti Vives}
\email{mvives1122@gmail.com}
\affiliation{Department of Electrical and Computer Engineering, Princeton University, Princeton, NJ 08544, USA}

\author{Mekena Metcalf}
\email{mekena.metcalf@gmail.com}
\affiliation{Applied Mathematics and Computational Research Division, Lawrence Berkeley National Laboratory, Berkeley, CA 94720, USA}

\begin{abstract}
Connecting superconducting qubits to optical fiber necessitates the conversion of microwave photons to optical photons.
Modern experimental demonstrations exhibit strong coupling between a microwave resonator and an optical cavity mediated through phononic modes in a mechanical oscillator. 
This paradigmatic transduction experiment is bounded by a theoretical efficiency with constant driving amplitudes on the electromagnetic resonators.
By adding a parametric drive to the microwave resonator and optical cavity we discover the converted signal through the quantum transducer is amplified, while maintaining a lower level of the added noise.
We propose a theoretical framework for time-dependent control of the driving lasers based on the input-output formalism of quantum optics, and solve analytically the 
transduction efficiency and added noise when the control signals parametrically drive the system. Our results show better transduction efficiency and lower added noise in varying parameter regimes relevant to current transduction experiments.

\end{abstract}

\maketitle

\section{Introduction}\label{sec:intro}
Quantum interconnect, a device which serves as
a coherent interface between otherwise incompatible physical media, is a critical component for the quantum network~\cite{QuantumIC_2021,cirac1997quantum, Mirhosseini2020-dq, Lauk2020-sx}, and viewed by many as essential for any hybrid architecture that may illuminate a viable path to scalable quantum computers~\cite{Arute2019-fz, Zhong2020-sq, Figgatt2019-wt}. A quantum transducer is an  interconnect device used to connect qubits at disparate energy scales. Experimental demonstrations based on optomechanical systems validate the promise of quantum transduction, yet further progress is needed to achieve unity efficiency without introducing excessive noise~\cite{andrews2014bidirectional, higginbotham2018harnessing, aspelmeyer2014cavity,Zeuthen2020-sr}. A major challenge is that an ideal quantum transducer requires both lossless cavities and the impedance-matching condition~\cite{Safavi-Naeini2011-ox}, which are experimentally demanding.

To overcome the aforementioned challenge, methods based on teleportation~\cite{Barzanjeh2011entangling,Zhong2020-iy,Wu2021-ry}, single-mode squeezing~\cite{Lau2020-si,zhong2022quantum}, adaptive control~\cite{Zhang2018-yk} and interference and phase-sensitive amplification~\cite{Lau2019-tc} have been proposed. However, all of the above approaches demand additional resources such as classical channels, new experimental components or squeezed input states. 

In this paper, we formulate a control scheme to improve the transducer performance that, unlike other approaches, does not require any hardware redesign. It needs only the ability to modulate the pump lasers on the electromagnetic (EM) resonators at twice the mechanical frequency. Such modulation is known to generate two-mode squeezing between the cavity photon and the mechanical phonon and enhance their entanglement~\cite{Mari2009-jx}. We will call our control protocol as parametric driving (PD) throughout this paper. Note the technique of parametrically driving a quantum transducer has been discussed in the context of single-mode squeezing~\cite{Lau2020-si,zhong2022quantum}. However, additional components, such as the Josephson parametric amplifier, are needed to implement the singe-mode squeezing operators.

We demonstrate the efficacy of our protocol by providing an analytic solution to the transfer matrix. Such solution is obtained by solving the quantum Langevin equation~\cite{aspelmeyer2014cavity,Lecocq2016-rs,Tian2010}, and is only accessible in special cases with time-dependent drive signals. In the case of PD \textit{we derive an analytic solution for the steady state transfer matrix} and discover that PD, compared to the standard constant control protocol, leads to an amplification of transduction efficiency while maintaining a lower level of added noise. The low added noise is achieved through the suppression of the two-mode-squeezing interaction (TMSI), similar to Ref.~\cite{Lau2020-si}. This work and Ref.~\cite{zhong2022quantum} are the first to report PD induced transduction efficiency enhancement.

Additionally, we use our theoretical tools to analyze an existing experimental implementation and readily show an improvement over constant cavity driving. Hence our PD strategy is a valuable addition, and complements the existing protocols and error suppression techniques~\cite{higginbotham2018harnessing, Zhang2018-yk,Lau2020-si} to further improve transduction performance.

The structure of this paper is as follows. In Sec.~\ref{sec:time_dependent_opto} we define the optomechanical transducer Hamiltonian and derive a general framework to include time-dependent controls based on the input-output theory of quantum optics. Next we apply this general theoretical framework to the parametrically driven quantum transducer. In this case, the transfer matrices are analytically solvable in the frequency domain. We present our main results regarding the properties of the transfer matrix in Sec.~\ref{sec:param_driven_sys}, and provide the corresponding proofs in Sec.~\ref{sec:proof}. Then based on the analytical solution, we examine the transduction efficiency and added noise of an ideal parametrically driven quantum transducer in Sec.~\ref{sec:ideal_transducer}, and a realistic one in Sec.~\ref{sec:realistic_transducer}. We show that, in both of the aforementioned cases, PD improves the transduction efficiency without amplifying the vacuum noise. We conclude in Sec.~\ref{sec:conclusion}, and present additional technical details in the Appendices.
\section{Time-Dependent Optomechanics}
\label{sec:time_dependent_opto} 
We start by formulating the Hamiltonian of the optomechanical transducer illustrated in Fig~\ref{fig:schematic}. We assume the two strong pump tones applied to both EM resonators have time-dependent amplitudes $\Omega_i\pqty{t}$ which serve as control signals. The Hamiltonian (in the laser frame of the pump tones) is
\begin{align}
    \label{eq:H}
    \hat{H}\pqty{t} &=  \sum_{i\in \Bqty{\ee,\oo}} \Delta_i \hat{a}_i^{\dagger}\hat{a}_i + g_i\hat{x}_\mm \hat{a}_i^{\dagger}\hat{a}_i + \Omega_i\pqty{t}\Ad{i} + \Omega^*_i\pqty{t}\A{i} \notag \\
    & \quad+ w_\mm \Ad{\mm}\A{\mm} \ .    
\end{align}
where $\A{i}$, $\Ad{i}$ are the annihilation and creation operators, and $\hat{x}_i = \Ad{i}+\A{i}$. We also set $\hbar = 1$. The subscript $\oo$ (optical) and $\ee$ (electrical) stand for the EM modes at optical and microwave frequencies while $\mm$ (mechanical) stands for the vibrational mode. We will call all of them cavities henceforth. $\Delta_i=\omega_i-\omega_{il}$ is the cavity ($\omega_i$) detuning with respect to the laser frequencies ($\omega_{il}$), and $g_\oo$ and $g_\ee$ are the optomechanical single-photon coupling strengths between the EM and mechanical cavities. 

\begin{figure}
    \centering
    \includegraphics[width=0.35\textwidth]{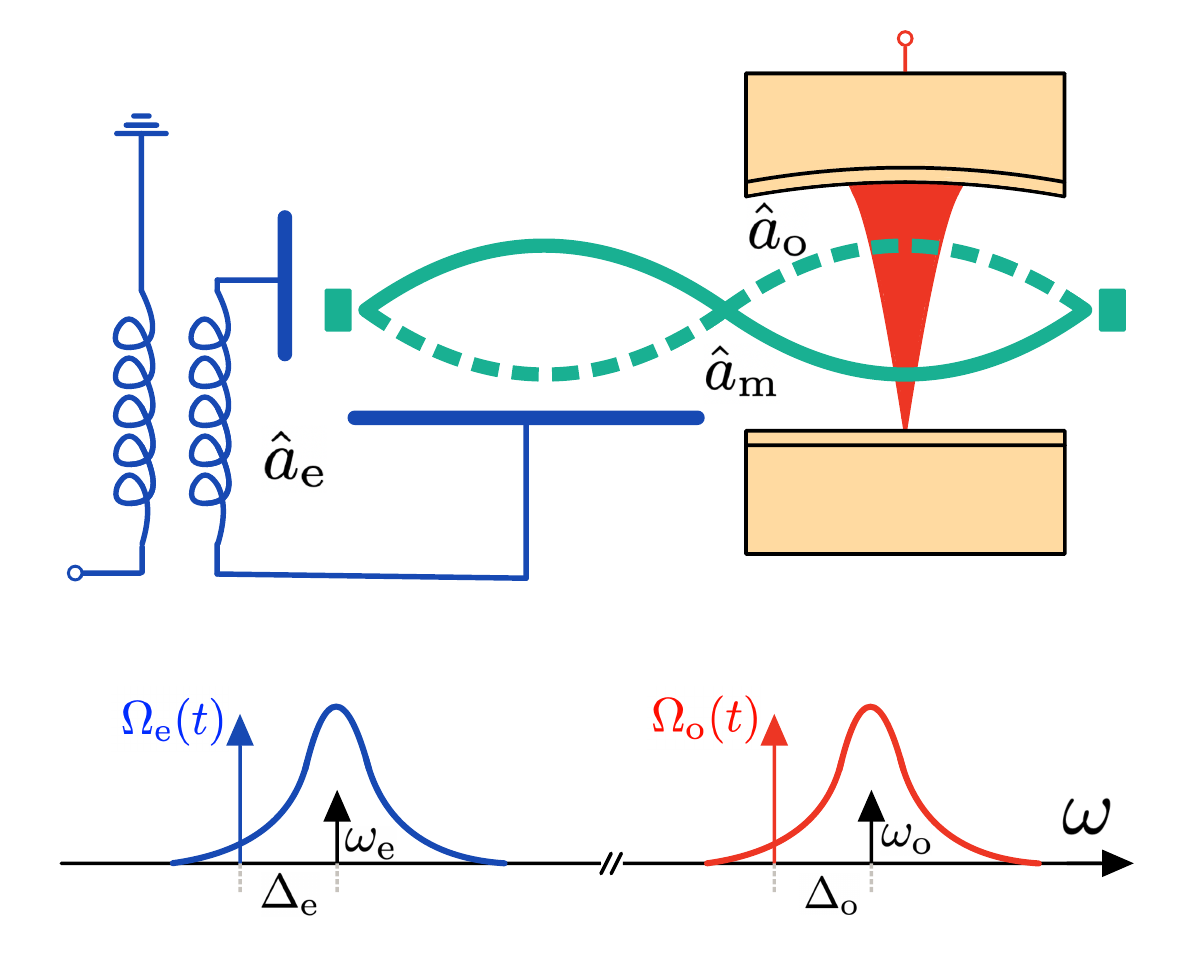}
    \caption{Schematic of an optomechanical quantum transducer. In this setup, an optical cavity ($\A{\oo}$) and a microwave resonator ($\A{\ee}$) couple to the same mechanical membrane ($\A{\mm}$). The control signals ($\Omega_i\pqty{t}$) are applied to both of the electromagnetic modes via pump lasers which are red detuned (by $\Delta_i$) from the corresponding cavity frequencies ($\omega_i$).}
    \label{fig:schematic}
\end{figure}

We derive an effective Hamiltonian with time-dependent controls following~\cite{machnes2012pulsed} (see Appendices~\ref{app:interaction_hamiltonian} and~\ref{app:geometric_D}). At the core of this approach is a perturbation theory that treats each EM cavity semiclassically and works with the perturbation operators $\delta\A{i} = \A{i}-\alpha_i\pqty{t}$ around the average values $\alpha_i(t)=\expval{\A{i}(t)}$ which are determined by
\begin{equation}
    \label{eq:a_eq_disspation}
    \dot{\alpha}_i(t) = -i\Delta_i\alpha_i(t)-\kappa_i\alpha_i(t)/2-i\Omega_i(t), \quad i\in\Bqty{\ee,\oo}\ ,
\end{equation}
where $\kappa_i$ are the total energy decay rates of cavity $i$.
This is permissible because the pump lasers are strong enough to dominate the cavity dynamics, i.e., $\vert\Omega_i(t)\vert \gg \Delta_i,\ g_i$. 
Consequently the perturbation operators are still small while $\Omega_i\pqty{t}$ take large values. 
The proturbation procedure consists of first rotating the Hamiltonian Eq.~\eqref{eq:H} w.r.t. the displacement operator $\mathcal{\hat{D}} = \otimes_{i\in\Bqty{\ee,\oo}}\D[i]{\alpha_i(t)}$, where $\D[i]{\alpha_i(t)}=\exp{\alpha_i(t)\Ad{i}-\alpha^*_i(t)\A{i}}$, and then ignoring the second order term proportional to $\delta \Ad{i} \delta \A{i}$. The later step is also known as linearization in the literature. For simplicity, we will use the notation $\A{i}$ for $\delta \A{i}$ from this point forward. The linearized Hamiltonian is
\begin{subequations}\label{eq:H_lin}
\begin{align}
    \hat{H}_{\mathrm{lin}}\pqty{t} &= \omega_\mm \Ad{\mm} \A{\mm} + \sum_{i\in\Bqty{\ee,\oo}} \Delta_i \Ad{i} \A{i} + g_i|\alpha_i(t)|^2 \hat{x}_\mm\\
    &+(G^*_i(t) \A{i} + G_i(t) \Ad{i}) \hat{x}_\mm \label{eq:H_inter} \ ,
\end{align}
\end{subequations}
where $G_i(t) \equiv g_i\alpha_i\pqty{t}$ are the effective EM-mechanical coupling strengths, which can be controlled by $\Omega_i\pqty{t}$. The purpose of pump lasers is to boost  the EM-mechanical coupling strengths by factors of $\alpha_i\pqty{t}$, so they are strong enough for signal transduction.

Based on the input-output formulation of cavities~\cite{Gardiner1985-dl,Walls2008-ks}, the quantum Langevin equation is derived by plugging Eq.~\eqref{eq:H_lin} into
\begin{equation}
    \label{eq:multi_input}
    \dot{\hat{a}}_i=i\comm{\hat{H}_\mathrm{lin}}{\A{i}} - \frac{\kappa_i}{2}\A{i}+\sum_j\sqrt{\kappa_{ij}}\A{ij,\mathrm{in}},\  i\in\Bqty{\ee,\mm,\oo} \ ,
\end{equation}
where $\A{ij,\mathrm{in}}$ is the $j$th input field operator (also referred to as the input) on cavity $i$, and $\kappa_{ij}$ is the corresponding coupling strength (see Appendix~\ref{app:input_output} for details). We also omitted the time-dependence of the operators. The input field operators directly model the photons injected from any coupling port
(such as the input mirror) into the cavity.
A compact form of Eq.~\eqref{eq:multi_input} is given by (see Appendices~\ref{app:input_output} and~\ref{app:eom} for details)
\begin{equation}
    \label{eq:eom_matrix_form}
    \mathbf{\dot{a}}(t) = \Ab(t)\mathbf{a}(t) + \Bb\mathbf{a}_{\mathrm{in}}\pqty{t} + \mathbf{v}\pqty{t} \ ,
\end{equation}
where
$
    \mathbf{a}\pqty{t}\equiv\bqty{\A{\oo}\pqty{t}, \A{\ee}\pqty{t}, \Ad{\oo}\pqty{t}, \Ad{\ee}\pqty{t}, \A{\mm}\pqty{t}, \Ad{\mm}\pqty{t}}^\mathrm{T}
$
is the state vector consisting of all the Heisenberg picture field operators, and $\mathbf{a}_{\mathrm{in}}\pqty{t}$ is the input vector consisting of all the inputs. The explicit expressions of $\Ab$, $\Bb$ and $\mathbf{v}\pqty{t}$ are given in Eqs.~\eqref{eq:ABv}.

In the main text, we focus on the following input vector
\begin{align}
    \label{eq:input_vector}
    \mathbf{a}_{\mathrm{in}} &\equiv \Big[ \A{\oo,\iin},\A{\oo \mathrm{I},\iin}, \A{\ee,\iin},\A{\ee \mathrm{I},\iin}, \Ad{\oo, \iin},\Ad{\oo\mathrm{I}, \iin}, \Ad{\ee,\iin},\Ad{\ee\mathrm{I},\iin},\notag\\
    &\quad\A{\mm,\iin}, \Ad{\mm,\iin}\Big]^{\mathrm{T}} \ ,
\end{align}
where $\A{i,\iin}$ ($i\in\Bqty{\ee,\oo}$) are the external inputs which model the incoming signal photons on cavity $i$, $\A{i\mathrm{I},\iin}$ and $\A{\mm,\iin}$ are the internal inputs which model the incident noise on cavities $i$ and $\mm$. We use the notation $\kappa_{i,\mathrm{ex}}$ for the coupling strengths to the external inputs. We also assume all the inputs are vacuum noise unless otherwise specified. Before proceeding, we emphasize that our results are true for an arbitrary configuration of inputs. 

The output vector can be defined similarly, and is connected to the input vector via
\begin{equation}
    \label{eq:input-output}
    \mathbf{a}_{\out}\pqty{t} + \mathbf{a}_{\iin}\pqty{t} = \Bb^\mathrm{T} \mathbf{a}\pqty{t} \ .
\end{equation}
The goal of this quantum transducer is to convert one photon from the external input of the electrical resonator to the external output of the optical resonator, or vice versa.
\begin{widetext}
    \begin{subequations}
        \label{eq:ABv}
        \begin{align}
        \Ab(t) &=  
            \begin{pmatrix}
                -i\Delta_\oo - \frac{\kappa_\oo}{2} & 0 & 0 & 0&-iG_\oo\pqty{t}&-iG_\oo\pqty{t}  \\
                0 & -i \Delta_\ee - \frac{\kappa_\ee}{2} & 0 & 0&-i G_\ee\pqty{t}&-i G_\ee\pqty{t}  \\
                0&0&i\Delta_\oo-\frac{\kappa_\oo}{2}&0&iG^*_\oo\pqty{t}&iG^*_\oo\pqty{t}\\
                0&0&0&i\Delta_e-\frac{\kappa_\ee}{2}&iG^*_\ee\pqty{t}&iG^*_\ee\pqty{t}  \\
                -iG^*_\oo\pqty{t} & -i G^*_\ee\pqty{t} & -iG_\oo\pqty{t} &-iG_\ee\pqty{t}&-i\omega_\mm-\frac{\kappa_\mm}{2}& 0  \\
                iG^*_\oo\pqty{t} & i G^*_\ee\pqty{t} & iG_\oo\pqty{t} &iG_\ee\pqty{t}&0& i\omega_\mm-\frac{\kappa_\mm}{2} 
            \end{pmatrix} \label{eq:At}\\
            \Bb &=\begin{pmatrix}
                \mathbf{D}& 0 & 0 \\
                0 & \mathbf{D} & 0\\
                0 & 0 & \mathbf{M}\\
            \end{pmatrix}, \quad 
            \mathbf{D} =\begin{pmatrix}
                \sqrt{\kappa_{\oo 1}} &\cdots& \sqrt{\kappa_{\oo k_{\oo}}} & 0 &\cdots & 0  \\
                0 & \cdots & 0 & \sqrt{\kappa_{\ee 1}} &\cdots& \sqrt{\kappa_{\ee k_{\ee}} }
            \end{pmatrix} \label{eq:Bmat}\\
                \mathbf{M} &=\begin{pmatrix}
                \sqrt{\kappa_{\mm 1}} &\cdots& \sqrt{\kappa_{\mm k_{\mm}}} & 0 &\cdots & 0  \\
                0 & \cdots & 0 & \sqrt{\kappa_{\mm 1}} &\cdots& \sqrt{\kappa_{\mm k_{\mm}} }
            \end{pmatrix}\\
            \mathbf{v}\pqty{t} &= \bigg[0,
            0,0,
            0,-ig_\oo\abs{\alpha_\oo\pqty{t}}^2 - ig_\oo\abs{\alpha_\ee\pqty{t}}^2,ig_\oo\abs{\alpha_\oo\pqty{t}}^2 + ig_\ee\abs{\alpha_\ee\pqty{t}}^2
            \bigg]^\mathrm{T}\ .\label{eq:v}
        \end{align}
        \end{subequations}
\end{widetext}

\section{Parametrically Driven System}
\label{sec:param_driven_sys}
We can improve upon the standard practice of using constant control signals ($\Omega_i\pqty{t}=\Omega_i$) by choosing oscillating control signals at twice the input signal frequency $ \Omega_i\pqty{t} = \Omega_i e^{-2i\omega_\mm t}$. In this paper we consider only the case of symmetric driving amplitudes $\Omega_i=\Omega$.
We provide the main results here and a detailed derivation is found in Sec.~\ref{sec:proof}.

In the long time limit, the steady state solution of Eq.~\eqref{eq:a_eq_disspation} subject to the oscillating control has only a single-frequency component
\begin{equation}
    \label{eq:steady_alpha}
    \alpha^\mathrm{s}_i\pqty{t} = \frac{2\Omega_i}{4\omega_\mm-2\Delta_i+i\kappa_i}e^{-2i\omega_\mm t} \equiv \Omega^{\mathrm{s}}_ie^{-2i\omega_\mm t}\ .
\end{equation}
As a result, the steady state $\mathbf{A}(t)$ can be decomposed into a Fourier series
\begin{equation}
    \label{eq:AF}
    \Ab\pqty{t} = \AD + \Am e^{-2i\omega_\mm t} + \Ap e^{2i\omega_\mm t} \ ,
\end{equation}
where $\AD$, $\Am$ and $\Ap$ are constant matrices given by $\AD = \mathrm{diag}\pqty{\Ab\pqty{t}}$ and
\begin{equation}\label{eq:Apm}
    \Am =     \begin{pmatrix}
        0 & 0 & \mathbf{Q}_{\mathrm{am}}   \\
        0 & 0 & 0   \\
        0 & \mathbf{Q}_\mathrm{mc} & 0
    \end{pmatrix},\  \Ap =  
    \begin{pmatrix}
        0 & 0 & 0  \\
        0 & 0 & \mathbf{Q}_\mathrm{cm}  \\
        \mathbf{Q}_\mathrm{ma} & 0 & 0 
    \end{pmatrix} \ ,
\end{equation}
where 
\begin{equation}
    \mathbf{Q}_{\mathrm{am}} = \begin{pmatrix}
        -iG^\mathrm{s}_\oo & -iG^\mathrm{s}_\oo \\
        -iG^\mathrm{s}_\ee & -iG^\mathrm{s}_\ee 
    \end{pmatrix},\  \mathbf{Q}_\mathrm{mc} =  
    \begin{pmatrix}
        -iG^\mathrm{s}_\oo & -iG^\mathrm{s}_\ee \\
        iG^\mathrm{s}_\oo & iG^\mathrm{s}_\ee
    \end{pmatrix} \ ,
\end{equation}
and $\mathbf{Q}_{\mathrm{cm}}=\mathbf{Q}^*_{\mathrm{am}}$, $\mathbf{Q}_{\mathrm{ma}}=-\mathbf{Q}^*_{\mathrm{mc}}$, $G^\mathrm{s}_i=g_i \Omega^\mathrm{s}_i$.
It is convenient to examine the equation of motion in Fourier space by taking the Fourier transform on both sides of Eq.~\eqref{eq:eom_matrix_form}~\cite{Fourier}
\begin{align}
    \label{eq:shifted_eoq}
    -i\omega \mathbf{a}\pqty{\omega} &= \AD\mathbf{a}\pqty{\omega}+\Am\mathbf{a}\pqty{\omega-2\omega_\mm} \notag\\
    &\quad+\Ap\mathbf{a}\pqty{\omega+2\omega_\mm} + \Bb\mathbf{a}_{\mathrm{in}}\pqty{\omega} + \mathbf{v}\pqty{\omega} \ .
\end{align}
We observe that the PD introduces sidebands shifted by $\pm2\omega_\mm$. While solving Eq.~\eqref{eq:shifted_eoq} is difficult, we can proceed with a few approximations. First, we ignore the vector $\mathbf{v}\pqty{\omega}$ because its steady state only contributes to the direct current component. Second, we assume the number of sidebands generated by Eq.~\eqref{eq:shifted_eoq} can be truncated. Keeping the nearest $2N$ sidebands, 
Eq.~\eqref{eq:shifted_eoq} can be rewritten as
\begin{equation}
    \label{eq:expanded_eom}
     -i\mathbf{\bar{W}}\mathbf{\bar{a}}\pqty{\omega}=\mathbf{\bar{A}}\mathbf{\bar{a}}\pqty{\omega}+\mathbf{\bar{B}}\mathbf{\bar{a}}_{\mathrm{in}}\pqty{\omega}
\end{equation}
with a sidebands extended state vector $\mathbf{\bar{a}}\pqty{\omega}\equiv\bqty{\mathbf{a}^T\pqty{\omega-2N\omega_\mm},\dots\mathbf{a}^T\pqty{\omega},\dots\mathbf{a}^T\pqty{\omega+2N\omega_\mm}}^T$ and a similarly defined $\mathbf{\bar{a}}_{\mathrm{in}}\pqty{\omega}$. Additionally, $\mathbf{\bar{W}}$ is a diagonal matrix with the corresponding shifted frequencies and other quantities are
\begin{equation}
    \label{eq:Abar}
    \mathbf{\bar{A}}=\begin{pmatrix}
    \AD&\Ap&0&\dots&0\\
    \Am&\AD&\Ap&&0\\
    0&\Am&\AD&&\vdots\\
    \vdots&&&\ddots&\Ap\\
    0&0&\cdots&\Am&\AD
    \end{pmatrix} \ ,
\end{equation}
and $ \mathbf{\bar{B}}=\mathrm{diag}\Big(\Bb,\ \Bb, \dots,\ \Bb \Big)$.

Using Eqs.~\eqref{eq:expanded_eom} and~\eqref{eq:input-output} we can solve for the output vector in terms of the input vector. The solution, without of loss of generality, can be written in the scattering form
\begin{equation}
    \label{eq:general_transfer_T}
    \mathbf{a}_{\mathrm{out}}\pqty{\w}=\TT\pqty{\w}\mathbf{a}_{\mathrm{in}}\pqty{\w} + \sum_{i=1}^{N}\TT_\pm^{\bqty{k}}\pqty{\w}\mathbf{a}_{\iin}\pqty{\w\pm 2k\w_\mm} \ ,
\end{equation}
where $\TT\pqty{\w}$ and $\TT_\pm^{\bqty{k}}\pqty{\w}$ are known as the transfer matrices and their explicit forms are given in Sec.~\ref{sec:proof}. Eq.~\eqref{eq:general_transfer_T} suggests that inputs at other sidebands are mixed into the outputs via $\TT_\pm^{\bqty{k}}\pqty{\w}$. We use the scattering notation $ \A{i,\out}\pqty{\w} = \TT_{i,j'}\pqty{\w}\Ad{j,\iin}\pqty{\w}$ to denote one element of the transfer matrix where the prime symbol in the subscript means the corresponding element is an annihilation operator.

Given the general form of Eq.~\eqref{eq:general_transfer_T}, we present the two main results of this paper. The first one states that $\TT\pqty{\w}$ has a block diagonal structure
$
    \TT\pqty{\w} = \mathrm{diag}\bqty{\TT_\mathrm{a}\pqty{\w}, \TT_\mathrm{c}\pqty{\w},\TT_\mathrm{m}\pqty{\w}}
$
where $\TT_\mathrm{a}\pqty{\w}$ and $\TT_\mathrm{c}\pqty{\w}$ ($\TT_\mm\pqty{\w}$) are $4\times4$ ($2\times 2$) matrices. It has two implications. First, there are no conjugate transmissions in $\TT\pqty{\w}$, i.e., $\TT_{i,j'}\pqty{\w} = 0$ for any $i$ and $j$. This statement indicates that our protocol suppresses the TMSI, effectively reducing the added noise in the output signal. Second, any energy coming from the mechanical inputs at the central frequency will be trapped within the mechanical mode, inducing enhanced excitation of mechanical phonons. Such enhancement leads to an amplification of the transduction efficiency, provided that the system is stable.

The second result states that $\TT_\pm^{\bqty{k}}\pqty{\w}=0$ for $k>2$. Furthermore, for more than four sidebands ($N > 2$), $\mathbf{T}\pqty{\omega}$ and $\TT_\pm^{\bqty{k}}\pqty{\w}$ remain identical to the $N=2$ case, meaning sidebands beyond $N=2$ do not contribute to the transfer matrices at the central frequency. Additionally, by examining the structure of $\TT_\pm^{\bqty{1}}\pqty{\w}$, we discover that $\TT_\pm^{\bqty{1}}\pqty{\w}$ only couples the EM outputs to the mechanical inputs at the nearest lower sideband, and $\TT_\pm^{\bqty{2}}\pqty{\w}$ only couples the EM outputs to the conjugate components of the EM inputs at the second nearest lower sideband. For example, if we consider the output of the optical cavity, the only terms being mixed into $\A{\oo,\out}\pqty{\w}$ by $\TT_\pm^{\bqty{k}}\pqty{\w}$ are
\begin{align}
\label{eq:mixed_aout}
    &\A{\oo,\out}\pqty{\w} = \cdots +\sum_{i\in\Bqty{\oo,\ee,\oo \mathrm{I}, \ee \mathrm{I}}} \mathbf{V}_{\oo,i'}\Ad{i,\iin}\pqty{\w - 4\w_\mm} \notag\\&+ \mathbf{U}_{\mathrm{\oo,m}}\A{\mm,\iin}\pqty{\w-2\w_\mm}
    + \mathbf{U}_{\oo,\mm'}\Ad{\mm,\iin}\pqty{\w-2\w_\mm} \ ,
\end{align}
where $\mathbf{U}=\TT_-^{\bqty{1}}\pqty{\w}$ and $\mathbf{V}=\TT_+^{\bqty{2}}\pqty{\w}$. Eq.~\eqref{eq:mixed_aout} includes sideband noise sources being introduced by the PD protocol. They can be suppressed by squeezing the corresponding inputs because their frequencies are different from the signal's~\cite{Lau2020-si}. Eq.~\eqref{eq:mixed_aout} also indicates that our strategy trades the central frequency mechanical noise with the sideband noise. Whether such an observation leads to some error suppression strategy is subject to future studies. Finally, it is worth mentioning that mechanical annihilation operator at the negative frequency and the mechanical creation operator at positive frequency all equal zero, i.e., $\A{\mm,\iin}\pqty{\w}=0$ for $\omega < 0$ and $\Ad{\mm,\iin}\pqty{\w}=0$ for $\omega > 0$, because the mechanical operators are rotated by the laser frame (see Appendix~\ref{app:input_output}).

\section{Transfer matrix solution}\label{sec:proof}
In this section we prove the results presented in Sec.~\ref{sec:param_driven_sys}.  The
reader who is not interested in the technical details of the proof may skip ahead to Sec.~\ref{sec:ideal_transducer}.

We start by presenting a \textit{formal solution} of the transfer matrices in Eq.~\eqref{eq:general_transfer_T}
\begin{subequations}
    \label{eq:TT}
    \begin{align}
        \TT &= \pqty{\Bb^\mathrm{T} \mathbf{X} \Bb- \II} \\
        \TT_{\pm}^{\bqty{k}} &= \Bb^\mathrm{T}\mathbf{X}\prod_{i=1}^{k} \pqty{\Ab_\pm\XX{i}{\pm}}\mathbf{B} \ , \label{eq:Tpmk}
    \end{align}
\end{subequations}
where
$
    \mathbf{X} = \pqty{-i\omega\II - \AD-\Xm^{\bqty{1}}-\Xp^{\bqty{1}}}^{-1}
$
and $\Xpm^{\bqty{k}}$, $\XX{k}{\pm}$ can be obtained using the recursive relation
\begin{subequations}
\label{eq:recursive}
\begin{align}
    \XX{k}{\pm}&=\bqty{-i\pqty{\omega \pm 2k\omega_\mm}\II-\AD-\Xpm^{\bqty{k+1}}}^{-1} \label{eq:Xinv}\\
    \Xpm^{\bqty{k+1}}&=\Apm \XX{k+1}{\pm} \Amp \label{eq:AXA}
\end{align}
\end{subequations}
and the boundary condition
\begin{equation}
    \label{eq:boundary}
    \XX{N}{\pm} = \bqty{-i\pqty{\omega\pm2N\omega_\mm}\II-\AD}^{-1} \ .
\end{equation}
(We use the convention $\prod_{i=1}^{k}\mathbf{O}_i = \mathbf{O}_1 \mathbf{O}_2 \ldots$ for the $\prod$ symbol.) The above solution is obtained by iteratively solving each equation in Eq.~\eqref{eq:expanded_eom}, starting with the boundary equations (see Appendix~\ref{app:freq_osc} for details).

Next we present the results in Sec.~\ref{sec:param_driven_sys} as theorems and provide the corresponding proofs. Without loss of generality, only the proofs for the upper sidebands matrices (with $+$ in the subscript) are presented, and the same procedures trivially apply to the lower sidebands matrices. 

Given the formal solution, the key observation is given by the following lemma.
\begin{lemma}\label{lem:diag}
$\XX{k}{\pm}$ and $\Xpm^{\bqty{k}}$ are block diagonal matrices with $2 \times 2$ blocks.
\end{lemma}

\begin{proof}
We prove the lemma by induction. First, for $k=N$, it is straightforward to verify that $\XX{N}{+}$ is block diagonal from Eq.~\eqref{eq:boundary} because the inverse of a block diagonal matrix is also block diagonal, i.e., $\mathbf{D}^{-1}\equiv\diag\bqty{\mathbf{D}_1, \mathbf{D}_2, \mathbf{D}_3}^{-1} = \diag\bqty{\mathbf{D}^{-1}_1, \mathbf{D}^{-1}_2, \mathbf{D}^{-1}_3}$. Since the matrix on the right-hand side of Eq.~\eqref{eq:boundary} is diagonal, we are free to partition it into $2 \times 2$ submatrices.

Second, we show $\XX{k}{+}$, $\XI{k}{+}$ and $\XI{k+1}{+}$ are block diagonal if $\XX{k+1}{+}$ is block diagonal. From the definition of $\mathbf{A}_\pm$ in Eq.~\eqref{eq:Apm}, we find that $\Ap \mathbf{D} \Am$ preserves the block diagonal structure of $\mathbf{D}$ when $\mathbf{D}_i$ are $2 \times 2$ matrices
\begin{equation}\label{eq:ApDAm}
    \Ap \mathbf{D} \Am = \begin{pmatrix}
        0 & 0 & 0   \\
        0 & \mathbf{Q}_\mathrm{cm} \mathbf{D}_3 \mathbf{Q}_\mathrm{mc} & 0   \\
        0 & 0 &  \mathbf{Q}_\mathrm{am}\mathbf{D}_1 \mathbf{Q}_\mathrm{ma}
    \end{pmatrix} \ .
\end{equation}
Therefore, based on Eq.~\eqref{eq:AXA}, $\XI{k+1}{+}$ is block diagonal with $2 \times 2$ blocks. Then every term within the square brackets on the right-hand side of Eq.~\eqref{eq:Xinv} is block diagonal. Thus $\XX{k}{+}$ is block diagonal. Finally, using Eq.~\eqref{eq:ApDAm} and~\eqref{eq:AXA} again, we prove that $\XI{k}{+}$ is also block diagonal.
\end{proof}

Lemma~\ref{lem:diag} allows us to simplify the recursive relation using the block diagonal structure of $\XX{k}{\pm}$ and $\Xpm^{\bqty{k}}$, which opens the way for the following three theorems.

\begin{theorem}
$\mathbf{T}$ is block diagonal with block sizes $l\times l$, $l\times l$ and $2k_\mm\times 2k_\mm$, where $l=k_\ee + k_\oo$ and $k_i$ are the number of inputs on cavity $i$.
\end{theorem}

\begin{proof}
The proof of this theorem directly follows lemma \ref{lem:diag}. Because $\XI{1}{\pm}$ are block diagonal with $2 \times 2$ blocks, $\mathbf{X}$ is block diagonal with $2 \times 2$ blocks. Recall from Eq.~\eqref{eq:Bmat} that $\Bb$ is also a block diagonal matrix with block size $2 \times l$, $2 \times l$ and $2 \times 2k_\mm$. As a consequence, $\Bb^\mathrm{T} \mathbf{X} \Bb$, therefore $\mathbf{T}$, are block diagonal matrices with block sizes $l\times l$, $l\times l$ and $2k_\mm\times 2k_\mm$.
\end{proof}
\begin{theorem}\label{the:Tpm2}
$\TT_{\pm}^{\bqty{k}}=0$ for $k > 2$.
\end{theorem}

\begin{proof}
We prove this theorem by proving a sufficient condition: $\prod_{i=1}^{k}\pqty{\Ab_\pm\XX{k}{+}} = 0$ for all $ k > 2$. Recall from lemma~\ref{lem:diag} that $\XX{+}{k}$ are block diagonal. We use the notation $\XX{k}{i}$ to denote each of its blocks, i.e., $ \XX{k}{+} = \diag\bqty{\XX{k}{1}, \XX{k}{2}, \XX{k}{3}}$. Then it is straightforward to verify $\prod_{i=1}^{3}\pqty{\Am\XX{3}{+}} = 0$ by writing down the cases of $k = 1, 2$ explicitly
\begin{subequations}
\label{eq:sideband_input_upper}
\begin{align}
    \Ab_+\XX{1}{+} &= \begin{pmatrix}
        0 & 0 & 0   \\
        0 & 0 & \mathbf{Q}_\mathrm{cm} \XX{1}{3}  \\
        \mathbf{Q}_\mathrm{ma} \XX{1}{1} & 0 & 0
        \end{pmatrix} \\
    \Ab_+\XX{1}{+}\Ab_+\XX{2}{+} &= \begin{pmatrix}
        0 & 0 & 0 \\
        \mathbf{Q}_\mathrm{cm}\XX{1}{3}\mathbf{Q}_\mathrm{ma}\XX{2}{1}& 0 & 0 \\
        0 & 0 & 0
    \end{pmatrix} \ .
\end{align}
\end{subequations}
Therefore we have $\prod_{i=1}^{k}\pqty{\Am\XX{k}{+}} = 0$ for $k > 2$.
\end{proof}

\begin{theorem}
\label{the:indepdent_N}
$\mathbf{T}$, $\TT_{\pm}^{\bqty{1}}$ and $\TT_{\pm}^{\bqty{2}}$ are independent of $N$ for $N \ge 2$.
\end{theorem}

\begin{proof}
The strategy of this proof is as follows. First we rewrite the recursive relation in Eqs.~\eqref{eq:recursive} as relations of the corresponding submatrices. Then we show every submatrix of $\XI{1}{\pm}$, $\Ab_+\XX{1}{+}$ and $\Ab_+\XX{1}{+}\Ab_+\XX{2}{+}$ are indepdent of $N$ when $N \ge 2$. Finally, we prove $\mathbf{T}$, $\TT_{\pm}^{\bqty{1}}$ and $\TT_{\pm}^{\bqty{2}}$ are independent of $N$ by directly following Eqs.~\eqref{eq:TT}.

Similar to the proof of theorem~\ref{the:Tpm2}, we denote the diagonal blocks of $\mathbf{A}_\mathrm{d}$ and $\Xp^{\bqty{k}}$ by
\begin{subequations}
    \begin{align}
        \mathbf{A}_\mathrm{d} &= \diag\bqty{\mathbf{A}_\mathrm{d1}, \mathbf{A}_\mathrm{d2}, \mathbf{A}_\mathrm{d3}}\\
        \Xp^{\bqty{k}} &= \diag\bqty{\mathbf{\Xi}^{\bqty{k}}_{1}, \mathbf{\Xi}^{\bqty{k}}_{2}, \mathbf{\Xi}^{\bqty{k}}_{3}} \ .
    \end{align}
\end{subequations}
Then from Eqs.~\eqref{eq:AXA}  and~\eqref{eq:ApDAm} we have
\begin{equation}
    \label{eq:Xi1Xi3}
    \XI{k}{1} = 0, \quad \XI{k}{3} = \mathbf{Q}_\mathrm{am} \XX{k}{1} \mathbf{Q}_\mathrm{ma} \ ,
\end{equation}
where
\begin{equation}
    \XX{k}{1} = \pqty{-i\pqty{\w + 2k\omega_\mm}\mathbf{I}-\mathbf{A}_\mathrm{d1}}^{-1}
\end{equation}
according to Eq.~\eqref{eq:Xinv} ($\XI{k+1}{1}$ is also zero from Eq.~\eqref{eq:Xi1Xi3}). The last element of $\Xp^{\bqty{k}}$ is given by
$
    \XI{k}{2} = \mathbf{Q}_\mathrm{am} \XX{k}{3}\mathbf{Q}_\mathrm{ma}
$
where
\begin{equation}\label{eq:X13}
    \XX{k}{3} = \pqty{-i\pqty{\w + 2k\omega_\mm}\mathbf{I}-\mathbf{A}_\mathrm{d3}-\XI{k+1}{3}}^{-1} \ .
\end{equation}
The key observation here is that $\XI{k}{1}$, $\XI{k}{3}$ and $\XX{k}{1}$ depend only on $k$; $\XI{k}{2}$ and $\XX{k}{3}$ depend only on $k$ and $k + 1$.

As a result, for $N \ge 2$, $\Xp^{\bqty{1}}$ (similarly $\Xm^{\bqty{1}}$), therefore $\mathbf{X}$, are independent of $N$. Additionally $\prod_{i=1}^{k} \pqty{\Ab_\pm\XX{i}{\pm}}$ are also independent of $N$ for $k=1,2$ based on Eqs.~\eqref{eq:sideband_input_upper}. Hence from Eqs.~\eqref{eq:TT}, we prove $\mathbf{T}$, $\TT_{\pm}^{\bqty{1}}$ and $\TT_{\pm}^{\bqty{2}}$ are independent of $N$ for $N \ge 2$.
\end{proof}

The above three theorems encapsulate all the results we presented in Sec.~\ref{sec:param_driven_sys}. Furthermore, theorem~\ref{the:indepdent_N} indicates that the exact solution of the transfer matrices can be found by examining only the cases of $N \leq 2$. Adding more sidebands will not change the result.
Lastly, by further examining the positions of non-zero blocks in Eqs.~\eqref{eq:sideband_input_upper}, we notice $\TT_+^{\bqty{1}}$ only couple the EM cavity operators to the mechanical cavity operators, while $\TT_+^{\bqty{2}}$ only couple the EM cavity operators to EM cavity operators.

\section{Ideal Transducer}\label{sec:ideal_transducer}
We apply the general theory to an ideal quantum transducer, which means the pump lasers are perfectly red detuned from the corresponding cavity frequencies by $\w_\mm$, i.e., $\Delta_\oo=\Delta_\ee=\w_\mm$ and the EM cavities are symmetric and lossless, i.e., $\kappa_{\ee,\mathrm{ex}}=\kappa_{\oo,\mathrm{ex}}=\kappa_\ee = \kappa_\oo=\kappa$ and $g_\ee=g_\oo=g$. We consider the scenario where $\A{\ee,\iin}$ is in the single-frequency coherent state with an average of one photon, and the goal of the transducer is to convert the microwave photon to an optical photon. The relevant figure of merits are the transduction efficiency and the added noise, both defined in the frequency domain. The former is defined as $\eta\pqty{\omega} \equiv \abs{\TT_{\oo,\ee}\pqty{\omega}}$, and
the latter is defined as~\cite{Caves1982-qx,Lau2020-si}
\begin{equation}
    2\pi\eta^2\pqty{\omega}S\pqty{\omega}\delta\pqty{\w-\w'}=\frac{1}{2}\expval{\acomm{\hat{J}\pqty{\w}}{\hat{J}\pqty{\w'}^\dagger}} \ ,
\end{equation}
where $\hat{J}\pqty{\w}$ is the noise terms being mixed into the output. 
Without loss of generality, we assume it has the form of
$\hat{J}\pqty{\w}=\TT_{\oo,\ee'}\pqty{\w}\Ad{\ee,\iin}\pqty{\w}+\sum_i (\TT_{\oo,i}\A{i,\iin}+\TT_{\oo,i'}\Ad{i,\iin})$ where the index $i$ now includes the frequency shifts and goes through all the inputs on the right-hand side of Eq.~\eqref{eq:general_transfer_T} except for $\A{\ee, \iin}\pqty{\w}$ and $\Ad{\ee, \iin}\pqty{\w}$. For example, terms like $\A{\oo,\iin}\pqty{\w}$, $\A{\oo \mathrm{I},\iin}\pqty{\w}$ and $\Ad{\oo,\iin}\pqty{\w-4\omega_\mm}$ are included by index $i$.
$S\pqty{\omega}$ can be calculated directly from the elements of the transfer matrices (see Appendix~\ref{app:added_noise} for details)
\begin{align}
\label{eq:added_noise_expression}
    \eta^2\pqty{\w}S\pqty{\w}&=\frac{3}{2}\abs{\TT_{\oo,\ee'}\pqty{\w}}^2 \notag\\ &+ \frac{1}{2}\sum_i \pqty{\abs{\TT_{\oo,i}}^2+\abs{\TT_{\oo,i'}}^2 } \ ,
\end{align}
and a lower bound is given by~\cite{Lau2020-si}
\begin{equation}
\label{eq:S_lower_bound}
S\pqty{\w} \ge \frac{3}{2}R^2\pqty{\w} + \abs{\frac{1-\eta^2\pqty{\w}}{2\eta^2\pqty{\w}}+\frac{R^2\pqty{\w}}{2}} \ ,
\end{equation}
where $R^2\pqty{\w}=\abs{\TT_{\oo,\ee'}\pqty{\w}}^2/\eta^2\pqty{\w}$. The lower bound states that there will be excessive noise if the transduction efficiency goes beyond unity. Before proceeding, we make a few additional comments. First, we need to set the corresponding transfer matrix elements to zero in Eq.~\eqref{eq:added_noise_expression} for operators which are zero, e.g., $\A{\mm, \iin}\pqty{\w}$ where $\w<0$. Second, Eq.~\eqref{eq:added_noise_expression} is only true when all the incident noises are vacuum (for a general case, see Appendix~\ref{app:added_noise}). Third, when PD is applied, the lower bound can be achieved by squeezing every noisy input~\cite{Lau2020-si}. However, we will focus on the raw performance metrics without any squeezing operation here. Lastly, because the input signal frequency is $\omega_\mm$, we will focus on $\eta\pqty{\omega_\mm}$ and $S\pqty{\w_\mm}$ and omit the frequency dependence henceforth.

First we consider the limit of zero mechanical loss $\kappa_\mm \to 0$. Analytic solutions of $\TT$ can be obtained for $N=1$ and $N=2$ (see Appendix~\ref{app:freq_osc} and~\cite{notebook}). The $N=1$ solution was discussed in the context of entanglement enhancement~\cite{Mari2009-jx}. In this case, the transduction efficiency goes to unity
$
    \eta=1
$
and all the unwanted noise is completely suppressed
$
\abs{\TT_{\oo,\ee'}}=\abs{\TT_{\oo,\oo'}}=\abs{\TT_{\oo,\oo}}=0
$.
For $N=2$, the transduction efficiency is
$\eta=\sqrt{1+\kappa^2/16\omega^2_\mm}$ and additional noise is introduced by the unwanted reflection and sideband coupling $\abs{\TT_{\oo, \oo}}=\abs{\mathbf{V}_{\oo,\ee}}=\abs{\mathbf{V}_{\oo,\ee'}}=\kappa/4\omega_\mm$. In both cases, PD offers noise suppression in the sideband unresolved limit $4\omega_\mm \lesssim \kappa$. Even though the $N=1$ solution is preferable, we do not identify a controllable parameter such that, when sent to zero, the solution would converge at $N=1$ in the sideband unresolved limit. Therefore the $N=1$ solution may not always be physically attainable. As demonstrated in Fig~\ref{fig:ideal}, the PD introduces less noise than the constant driving even when more sidebands are included.

Next, we show that for a transducer with small non-vanishing $\kappa_\mm$, the PD strategy enhances transduction efficiency while reducing the added noise. Such enhancement is caused by the additional energy trapped inside the mechanical mode. The main result is illustrated in Fig~\ref{fig:ideal}, where we plot the transduction efficiency and added noise of a symmetric quantum transducer with different $\kappa_\mm$ values. The other parameters are modified from a real world experiment~\cite{higginbotham2018harnessing}. In particular, $\Omega$ is chosen by approximately matching the effective coupling strength $\abs{G_i}$ to the reported experimental value. We observe that although more added noise is introduced to the output when $\kappa_\mm$ increases, the PD solutions always have less added noise than the constant one when the noise is below the classical limit. Within the regions of acceptable noise, 
the transduction efficiency is amplified by the PD from approximately $\kappa_\mm \gtrsim 1 (\mathrm{Hz})$. We also numerically confirm that for $N>2$, the results have converged to the $N=2$ case (see Appendix~\ref{app:more_sidebands}). Although it seems unnecessary to boost the transduction efficiency above unity, the reduced noise could offer more benefits in terms of the quantum channel capacity~\cite{zhong2022quantum,Pirandola2009-gl}. Additionally, we can fine-tune the efficiency to unity by deliberately impedance mismatching the system.
\begin{figure}[t]
    \centering
    \includegraphics[width=\columnwidth]{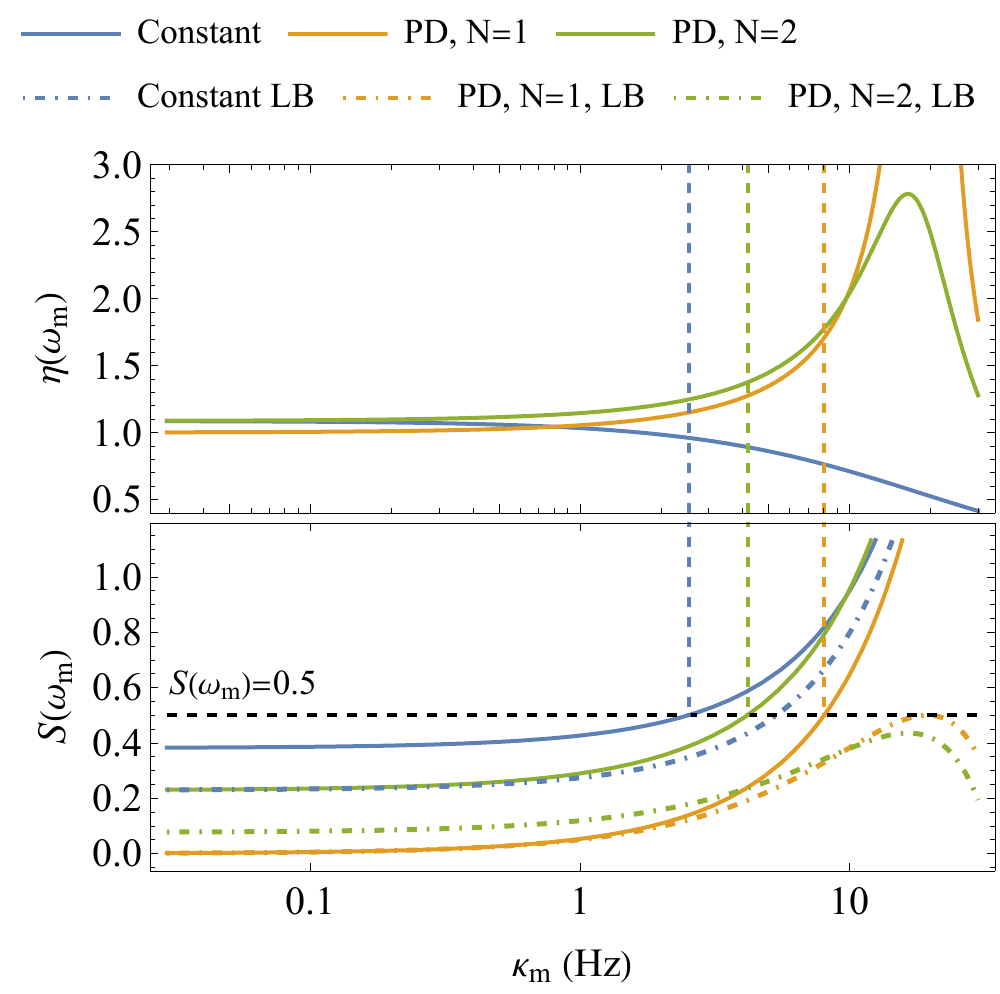}
    \caption{Transduction efficiency $\eta\pqty{\w_\mm}$ and added noise $S\pqty{\w_\mm}$ v.s. $\kappa_\mm$ of a symmetric transducer. The color dashed lines denote the $\kappa_\mm$ values where the added noise reaches the classical limit. Lower panel: the dash-dotted lines are the lower bounds (LBs) given by Eq.~\eqref{eq:S_lower_bound}; the black dashed line denotes the classical limit of added noise $S\pqty{\w_\mm}=0.5$. The parameter values used in the calculation are: $\omega_\mm/2\pi=1.4732\ \mathrm{MHz}$, $\kappa/2\pi=2.5\ \mathrm{MHz}$, $g/2\pi=3.8\ \mathrm{Hz}$ and $\Omega/2\pi = 500 \ \mathrm{MHz}$.
    }
    \label{fig:ideal}
\end{figure}

\section{Realistic Transducer}\label{sec:realistic_transducer}
This section demonstrates that the PD strategy is increasingly beneficial to a realistic optomechanical quantum transducer~\cite{higginbotham2018harnessing} than an ideal one by providing additional transduction efficiency amplification alongside noise reduction. Our analysis is based on the experimental configuration presented in~\cite{higginbotham2018harnessing}. Unfortunately, the experiment works in a noisy regime, i.e., $S\pqty{\w_\mm} > 0.5$, and the authors adopted a classical feed-forward protocol to reduce the added noise. While it is not clear whether a quantum feed-forward protocol exists, we still choose to work in the same parameter regime because our goal is to show that the PD strategy can readily improve the performance of an existing experimental transducer. 

Our results are reported in Fig~\ref{fig:exp} where we show the transduction efficiency and added noise as functions of the driving signal amplitude $\Omega$, whose range is chosen to keep the effective coupling strength $\abs{G_i}$ approximately around $\mathrm{kHz}$. From the upper panel, we clearly see that the PD offers a higher transduction efficiency than the constant driving throughout much of the observed region (starting from $\Omega \gtrsim 270(\mathrm{MHz})$), while the unity efficiency can be achieved by fine-tuning $\Omega$. From the lower panel, we observe that the added noise of the PD can be pushed below the level of the constant driving by increasing the driving amplitude. The PD starts exhibiting an advantage from $\Omega \gtrsim 512(\mathrm{MHz})$ for $N=1$ and $\Omega \gtrsim 636(\mathrm{MHz})$ for $N=2$. As with the ideal transducer, we also confirm that for $N > 2$, the results have converged to the $N=2$ case (see Appendix~\ref{app:more_sidebands}).

Lastly, we mention that the added noise is obtained by assuming all the noise inputs are in the vacuum state. In practice, the mechanical mode suffers from the optical-absorption-induced thermal noise~\cite{Ren2020-yd,MacCabe2020-yw}, which is the major bottleneck of optomechanical quantum transducers. On one hand, we emphasize that our results demonstrate that the PD can enhance the transduction efficiency without amplifying the vacuum noise, which is crucial for error cancellation strategies such as the feed-forward protocol~\cite{higginbotham2018harnessing}. Thus we expect the PD, combined with error cancellation techniques, to provide a viable path towards practical quantum transducers. On the other hand, we also expect the PD to offer some protection against the thermal noise if we detune the signal frequency from $\w_\mm$ because the PD decouples the mechanical inputs at the central frequency to the outputs. This possibility of thermal noise suppression is subject to future study.

\begin{figure}[t]
    \centering
    \includegraphics[width=\columnwidth]{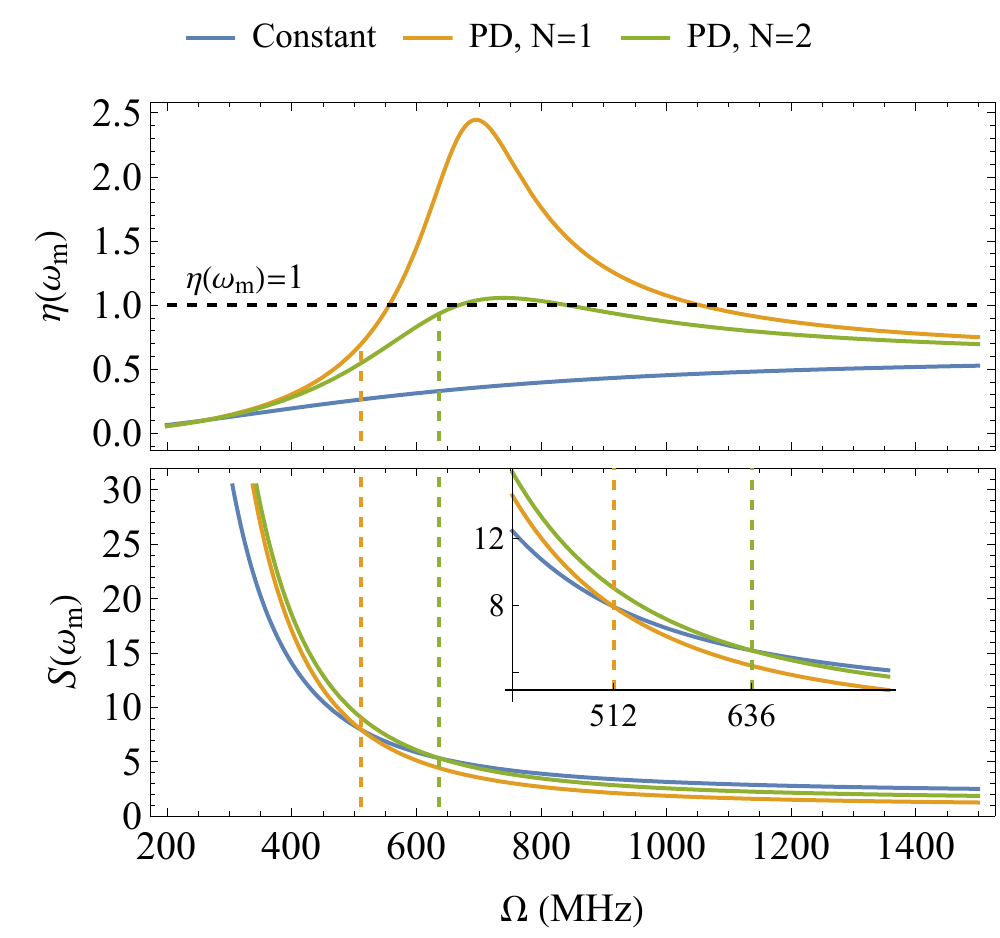}
    \caption{Transduction efficiency $\eta\pqty{\w_\mm}$ and added noise $S\pqty{\w_\mm}$ v.s. driving amplitude $\Omega$ of a realistic optomechanical quantum transducer ~\cite{higginbotham2018harnessing}. The color dashed lines denote the $\Omega$ value where the PD added noise falls below the constant one. Upper panel: the black dashed line denotes the unity transduction efficiency. Lower panel: the inset shows a zoomed-in version of the same plot. The numerical values used in the calculation are given in Table I of Appendix~\ref{app:num_vals}.}
     \label{fig:exp}
\end{figure}

\section{Conclusion and outlook}\label{sec:conclusion}
In this paper we present a theoretical analysis of the optomechanical quantum transducer with time-dependent controls. We demonstrate that, without amplifying the vacuum noise, better transduction efficiency can be achieved by replacing constant controls with oscillating ones. The improvement is applicable to existing experimental architectures. We hope this new development paves the way for more advanced time-dependent control protocols which, judging by the great success they have brought in designing quantum gates~\cite{Niu2019-tt, Werninghaus2021-bc, Baum2021-xg}, may lead to an even better transducer performance.

\begin{acknowledgments}
    The authors are grateful to Jie Luo, Zhi (Jackie) Yao, John Bell, Anastasiia Butko, Mariam Kiran and Wibe A. de Jong for useful discussions and feedback. This work was supported by the Laboratory Directed Research and Development Program of Lawrence Berkeley National Laboratory under U.S. Department of Energy Contract No. DE-AC02-05CH11231.
\end{acknowledgments}

\onecolumngrid
\appendix

\section{Interaction Picture Hamiltonian}
\label{app:interaction_hamiltonian}
In this section, we derive the interaction picture Hamiltonian of
\begin{equation}
    \hat{H}\pqty{t} =  \sum_{i\in \Bqty{\ee,\oo}} \Delta_i \hat{a}_i^{\dagger}\hat{a}_i + g_i\hat{x}_\mm \hat{a}_i^{\dagger}\hat{a}_i + \Omega_i\pqty{t}\Ad{i} + \Omega^*_i\pqty{t}\A{i} + w_\mm \Ad{\mm}\A{\mm}     
\end{equation}
in a frame rotating with the displacement operator $\mathcal{\hat{D}} = \otimes_i\D[i]{\alpha_i(t)}$, where $\D[i]{\alpha_i(t)}=\exp{\alpha_i(t)\Ad{i}-\alpha^*_i(t)\A{i}}$. Using the following relations
\begin{subequations}
\label{app:D_A_N}
\begin{gather}
    \Dd{\alpha}\A{}\D{\alpha} = \A{} + \alpha, \quad\quad \Dd{\alpha}\Ad{}\D{\alpha} = \Ad{} + \alpha^* \ ,\\
    \Dd{\alpha}\Ad{}\A{}\D{\alpha} = \Dd{\alpha}\Ad{}\D{\alpha}\Dd{\alpha}\A{}\D{\alpha}=\pqty{\Ad{} + \alpha^*}\pqty{\A{} + \alpha} \ ,
\end{gather}
\end{subequations}
we can write down the interaction picture Hamiltonian as
\begin{subequations}
\begin{align}
    &\hat{\tilde{H}} \equiv \mathcal{\hat{D}}^\dagger \hat{H} \mathcal{\hat{D}} -i\mathcal{\hat{D}}^\dagger\mathcal{\dot{\hat{D}}}=  \omega_\mm \Ad{\mm}\A{\mm} + \sum_{i=\ee,\oo}\Delta_i (\Ad{i}+\alpha^*_i)(\A{i}+\alpha_i) \\
    &\quad+g_i(\Ad{i}+\alpha^*_i)(\A{i}+\alpha_i)\hat{x}_\mm +\Omega_i(\Ad{i}+\alpha^*_i)+ \Omega^*_i(\A{i} + \alpha_i) \label{eq:G_interaction}\\
    &\quad +i\pqty{\dot{\alpha}_i^*\A{i}-\dot{\alpha}_i\Ad{i}} \label{eq:geometric}\ ,
\end{align}
\end{subequations}
where the time-dependence of $\alpha_i\pqty{t}$, $\Omega_i\pqty{t}$, $\alpha^*_i\pqty{t}$ and $\Omega^*_i\pqty{t}$ is omitted.
The geometric terms in line~\eqref{eq:geometric} (see Appendix~\ref{app:geometric_D} for details) can be further simplified by plugging
\begin{equation}
    \label{app:intra_cavity}
    \dot{\alpha}_i(t) = -i\Delta_i\alpha_i(t)-\frac{\kappa_i}{2}\alpha_i(t)-i\Omega_i(t) \ ,
\end{equation}
into the expression, which leads to
\begin{equation}
    \label{eq:geometric_simplified}
    i\pqty{\dot{\alpha}_i^*\A{i}-\dot{\alpha}_i\Ad{i}} = -\Delta_i\pqty{\alpha_i^*(t)\A{i}+\alpha_i(t)\Ad{i}}-i\frac{\kappa_i}{2}\alpha^*\pqty{t}\A{i}+i\frac{\kappa_i}{2}\alpha_i\pqty{t}\Ad{i} -\Omega_i\pqty{t}\Ad{i}-\Omega^*_i\pqty{t}\A{i} \ .
\end{equation}
Then we ignore the constant energy shift and linearize the Hamiltonian by ignoring the second order terms, i.e.,
\begin{equation}
    g_i(\Ad{i}+\alpha^*_i(t))(\A{i}+\alpha_i(t))\hat{x}_\mm \approx g_i(\alpha_i(t)\Ad{i}+\alpha_i^*(t)\A{i}) \hat{x}_\mm + g_i\abs{\alpha_i(t)}^2\hat{x}_\mm \ .
\end{equation}
Finally the time-dependent linearized Hamiltonian is given by
\begin{subequations}
\label{app:Hlin}
\begin{align}
    \hat{H}_{\mathrm{lin}} &= \omega_\mm \Ad{\mm} \A{\mm} + \sum_{i=\ee,\oo} \Delta_i \Ad{i} \A{i} +(G^*_i(t) \A{i} + G_i(t) \Ad{i}) \hat{x}_\mm + g_i|\alpha_i(t)|^2 \hat{x}_\mm \label{app:Hlin_s}\\
    &-i\frac{\kappa_i}{2}\alpha^*\pqty{t}\A{i}+i\frac{\kappa_i}{2}\alpha_i\pqty{t}\Ad{i} \label{app:kappa_Hlin}
\end{align}
\end{subequations}
with $G_i(t) = g_i\alpha_i\pqty{t}$. The open system effects $\kappa_i$ in line ~\eqref{app:kappa_Hlin} will be cancelled later in the quantum Langevin equation (discussed in Appendix~\ref{app:input_output}) so we will omit them when presenting the Hamiltonian.
\section{Geometric term generated by $\D{\alpha\pqty{t}}$}
\label{app:geometric_D}
In this section, we derive the geometric terms in line~\eqref{eq:geometric}. Because $\D{\alpha\pqty{t}}$ is a unitary matrix, it can be generated by a differential equation
\begin{equation}
    \label{app:geo_eq}
    \mathcal{\dot{D}}\pqty{t} = -i \hat{H}\pqty{t} \D{t}
\end{equation}
with an effective Hamiltonian
\begin{equation}
    \label{app:eff_H}
    \hat{H}\pqty{t}=i\bqty{\dot{\alpha}\pqty{t}\Ad{}-\dot{\alpha}^*\pqty{t}\A{}} +c\pqty{t}
\end{equation}
where $c\pqty{t}$ is a time-dependent c-number.
To prove this, we first show that the commutator $\comm{\hat{H}\pqty{t_1}}{\hat{H}\pqty{t_2}}$ is also a c-number
\begin{equation}
    \comm{\hat{H}\pqty{t_1}}{\hat{H}\pqty{t_2}} =- \comm{\dot{\alpha}\pqty{t_1}\Ad{}-\dot{\alpha}^*\pqty{t_1}\A{}}{\dot{\alpha}\pqty{t_2}\Ad{}-\dot{\alpha}^*\pqty{t_2}\A{}} = -\dot{\alpha}\pqty{t_1}\dot{\alpha}^*\pqty{t_2} + \dot{\alpha}^*\pqty{t_1}\dot{\alpha}\pqty{t_2} \ .
\end{equation}
Then we use the Magnus expansion to calculate the solution of Eq.~\eqref{app:geo_eq}
\begin{subequations}
\begin{align}
    T_+\exp{-i\int_0^\mathrm{T} \hat{H}\pqty{\tau}\dd{\tau}} &=\exp{-i\int_0^\mathrm{T}\hat{H}\pqty{\tau}\dd{\tau}-\frac{1}{2}\int_0^{t}\int_0^{t_1}\comm{\hat{H}\pqty{t_1}}{\hat{H}\pqty{t_2}}\dd{t_1}\dd{t_2}}\\
    &= \exp{\alpha\pqty{t}\Ad{}-\alpha^*\pqty{t}\A{}- \alpha\pqty{0}\Ad{}+\alpha^*\pqty{0}\A{}-f\pqty{t}} \label{app:geo_evo_f}\ .
\end{align}
\end{subequations}
If we choose $c\pqty{t}$ in Eq.~\eqref{app:eff_H} to be
\begin{equation}
    ic\pqty{t} = \frac{1}{2}\bqty{\dot{\alpha}\pqty{t}\alpha^*\pqty{t}-\dot{\alpha}^*\pqty{t}\alpha\pqty{t}} \ ,
\end{equation}
then $f\pqty{t}=0$ in Eq.~\eqref{app:geo_evo_f}.
It follows that the solution of Eq.~\eqref{app:geo_eq} is $\D{\alpha\pqty{t}}$:
\begin{equation}
    T_+\exp{-i\int_0^\mathrm{T} \hat{H}\pqty{\tau}\dd{\tau}} \hat{\mathcal{D}}\pqty{0}=\exp{\alpha\pqty{t}\Ad{}-\alpha^*\pqty{t}\A{}}\equiv \D{\alpha\pqty{t}} \ .
\end{equation}
As a result, the geometric term $-i\mathcal{\hat{D}}^\dagger\mathcal{\dot{\hat{D}}}$ becomes
\begin{equation}
    -i\mathcal{\hat{D}}^\dagger\mathcal{\dot{\hat{D}}} = -\mathcal{\hat{D}}^\dagger\hat{H}\pqty{t}\mathcal{\hat{D}} = i\bqty{\dot{\alpha}^*\pqty{t}\A{}-\dot{\alpha}\pqty{t}\Ad{}} \ ,
\end{equation}
where we ignored the terms that are proportional to identity.

\section{Input-output formalism in the interaction picture}
\label{app:input_output}
In this section, we re-derive the input-output formalism of quantum optics~\cite{Gardiner1985-dl} in the displacement frame. We assume each cavity mode in our configuration couples to a bosonic bath
$
    \hat{H}_\mathrm{B}=\sum_{i,k}\omega_{i,k}\Bd{i,k}\B{i,k} 
$
via the interaction Hamiltonian
\begin{equation}
    \label{eq:H_int}
    \hat{H}_{\mathrm{int}} = i\sum_{i,k} \gamma_{i, k}\pqty{\A{i}\Bd{i,k}e^{-i\omega_{li} t}-\B{i,k}e^{i \omega_{li}t}\Ad{i}}\ ,\quad i\in\bqty{\ee,\oo,\mm}
\end{equation}
in the interaction picture of the pump laser frequencies $\omega_{li}$. It follows $\omega_{l\mm}=0$ because there is no pump laser on the mechanical cavity.

After rotating the interaction Hamiltonian Eq.~\eqref{eq:H_int} with respect to the displacement operator $\mathcal{\hat{D}} = \otimes_i\D[i]{\alpha_i(t)}$ by replacing the annihilation and creation operators with
\begin{equation}
    \A{i}\to \A{i}+\alpha_i\pqty{t} \ , \quad \Ad{i}\to \Ad{i} + \alpha^*_i\pqty{t} \ ,
\end{equation}
we follow the standard procedure~\cite{Gardiner1985-dl} to derive the equation of motion. We start by writing down the equation of motion for $\A{i}\pqty{t}$ and $\B{i,k}\pqty{t}$
\begin{subequations}
\begin{align}
    \dot{\hat{a}}_i &= i\comm{\hat{H}_\mathrm{lin}}{\A{i}\pqty{t}}-\sum_{k}\gamma_{i,k}e^{i\omega_{li}t}\hat{b}_{i,k}\pqty{t} \label{app:int_a}\\
    \dot{\hat{b}}_{i,k} &= -i\omega_k \hat{b}_{i,k}\pqty{t}+\gamma_{i,k}e^{-i\omega_{li}t}\pqty{\A{i}\pqty{t}+\alpha_i\pqty{t}} \label{app:int_b}
    \ ,
\end{align}
\end{subequations}
where we omit the time-dependence of all the operators on the LHS.
The formal solution to Eq.~\eqref{app:int_b} is
\begin{equation}
    \label{app:int_b_sol}
    \hat{b}_{i,k} = e^{-i \omega_k \pqty{t-t_0}}\hat{b}_{i,k}\pqty{0} + \gamma_{i,k}\int_{t_0}^\mathrm{T} e^{-i \omega_k \pqty{t-\tau}} \pqty{\A{i}\pqty{\tau}+\alpha_i\pqty{\tau}}\dd{\tau} \ ,
\end{equation}
where $\hat{b}_{i,k}\pqty{0}=\hat{b}_{i,k}$.
Plugging Eq.~\eqref{app:int_b_sol} into Eq.~\eqref{app:int_a}, we have
\begin{equation}
    \label{app:2nd_order_a}
    \dot{\hat{a}}_i = i\comm{H_\mathrm{lin}}{\A{i}\pqty{t}} - \sum_k \gamma_{i,k}e^{i\omega_{li}t}e^{-i\omega_k\pqty{t-t_0}}\B{i,k} - \sum_k\gamma^2_{i,k}e^{i\omega_{li}t}\int_{t_0}^\mathrm{T} e^{-i \omega_k \pqty{t-\tau}} \pqty{\A{i}\pqty{\tau}+\alpha_i\pqty{\tau}}e^{-i\omega_{li}\tau}\dd{\tau} \ .
\end{equation}
Then we make the Markov approximation and replace the summation with an integral $\sum_k\gamma^2_{i,k}\to \gamma^2_i \int \dd{\omega} $. The last term in Eq.~\eqref{app:2nd_order_a} becomes
\begin{subequations}
\begin{align}
    \gamma^2_ie^{i\omega_{li}t}\int_{t_0}^\mathrm{T} \int_{-\infty}^{\infty} e^{-i\omega\pqty{t-\tau}}\dd{\omega}\pqty{\Ad{i}\pqty{\tau}+\alpha_i\pqty{\tau}}e^{-i\omega_{li}\tau}\dd{\tau} &=2\pi\gamma^2_ie^{i\omega_{li}t}\int_{t_0}^\mathrm{T} \delta\pqty{t-\tau}\pqty{\Ad{i}\pqty{\tau}+\alpha_i\pqty{\tau}}e^{-i\omega_{li}\tau}\dd{\tau} \\
    &= \frac{\kappa_i}{2} \pqty{\A{i}\pqty{t}+\alpha\pqty{t}} \label{app:a_p_alpha}\ ,
\end{align}
\end{subequations}
where $\kappa_i=2\pi\gamma^2_i$ are the total energy decay rate of cavity $i$. We notice that $\frac{\kappa_i}{2}\alpha_i\pqty{t}$ in line~\eqref{app:a_p_alpha} would cancel the same term in line~\eqref{app:kappa_Hlin}, so we will ignore it in later discussion. Using the standard definition of the \emph{input field operator}, Eq.~\eqref{app:2nd_order_a} becomes
\begin{equation}
    \label{app:a_equation}
    \dot{\hat{a}}_i = i\comm{\hat{H}_\mathrm{lin}}{\A{i}\pqty{t}} - \frac{\kappa_i}{2}\A{i}\pqty{t} + \sqrt{\kappa_i}\A{i,\mathrm{in}}\pqty{t}\ ,
\end{equation}
where
\begin{equation}
    \label{app:input_form}
    \A{i,\mathrm{in}}\pqty{t} = -\frac{1}{\sqrt{2\pi}}\sum_k e^{-i\pqty{\omega_k-\omega_{li}}t+i\omega_k t_0}\B{i, k} \ .
\end{equation}
Without loss of generality, we will set $t_0=0$ henceforth. Eq.~\eqref{app:input_form} can be written in the continuous limit
\begin{equation}
    \A{i,\mathrm{in}}\pqty{t} = -\frac{1}{\sqrt{2\pi}}\int_{-\infty}^{\infty} e^{-i\omega  t}\B{i}\pqty{\w} \dd{\omega}
\end{equation}
by shifting the central frequency $\w=\w_k-\w_{li}$ and replacing the summation with integral. Its Fourier transform is given by
\begin{equation}
    \label{app:a_F}
    \A{i,\iin}\pqty{\w}=\int_{-\infty}^{\infty} \A{i,\iin}\pqty{t} e^{i \w t} \dd{t} =-\sqrt{2\pi}\B{i}\pqty{\w} \ .
\end{equation}
The commutation relations of the time-domain and frequency-domain operators are
\begin{equation}
    \comm{\A{i,\iin}\pqty{t}}{\Ad{i,\iin}\pqty{t'}} = \delta\pqty{t-t'},\quad \comm{\hat{b}_{i}\pqty{\w}}{\hat{b}^\dagger_{i}\pqty{w'}} = \delta\pqty{\w-\w'} \ .
\end{equation}

Before proceeding, we make three additional comments on the input field operator. First, we emphasize that the Fourier transform and the Hermitian conjugation do not commute, i.e., $\mathcal{F}\bqty{\A{i,\iin}\pqty{t}}^\dagger \neq \mathcal{F}\bqty{\Ad{i,\iin}\pqty{t}}$. In this paper, we use the notation
\begin{equation}
    \label{app:dagger_defination}
    \Ad{i,\iin}\pqty{\w} \equiv \mathcal{F}\bqty{\Ad{i,\iin}\pqty{t}}\pqty{\w}, \quad \A{i,\iin}\pqty{\w}^\dagger \equiv \Bqty{\mathcal{F}\bqty{\A{i,\iin}\pqty{t}}\pqty{\w}}^\dagger \ .
\end{equation}
It is straightforward to check $\A{i,\iin}\pqty{\w}^\dagger = \Ad{i,\iin}\pqty{-\w}$ from Eq.~\eqref{app:a_F}. Second, for the electrical and optical input field operators, the negative frequency should be interpreted as, in the lab frame, the frequency smaller than the corresponding pump laser frequency. However, because the mechanical mode is not rotated in the pump laser frame, its input field operator does not have any negative frequency component, i.e., $\A{\mm,\iin}\pqty{\w} = 0$ for $\w < 0$. Based on the first comment, we also have $\Ad{\mm,\iin}\pqty{\w}=0$ for $\w > 0$. Third, the following commutation relation is re-normalized by an additional $2\pi$ factor
\begin{equation}
     \label{app:input_normalization}
    \comm{\A{i,\iin}\pqty{\w}}{\A{i,\iin}\pqty{\w'}^\dagger}=2\pi \delta\pqty{\w-\w'} \ .
\end{equation}

We note that the input field operator can be split into multiple parts depending on how we group the bath modes
\begin{equation}
    \label{app:multi-input}
    \dot{\hat{a}}_i=i\comm{\hat{H}_\mathrm{S}}{\A{i}} - \frac{\kappa_i}{2}\A{i}+\sum_k\sqrt{\kappa_{ik}}\A{ik,\mathrm{in}} \ ,
\end{equation}
where $\kappa_i = \sum_k \kappa_{ik}$. Eq.~\eqref{app:multi-input} is usually referred to as the \emph{quantum Langevin equation} in the literature~\cite{Walls2008-ks}.
The same procedure can be applied to derive a similar equation with output field operators
\begin{equation}
    \label{app:multi-output}
    \dot{\hat{a}}_i=i\comm{\hat{H}_\mathrm{S}}{\A{i}} + \frac{\kappa_i}{2}\A{i}-\sum_k\sqrt{\kappa_{ik}}\A{ik,\mathrm{out}} \ ,
\end{equation}
and the resulting equation is known as the \emph{time-reversed quantum Langevin equation}. The input/output field operators are often called \emph{input/output (ports)} in the literature in analogy to network theory. $\kappa_{ik}$ are called the \emph{coupling strengths} to the $k$th input/output of cavity $i$ and the summation of the coupling strengths to all the inputs/outputs equals the total energy decaying rate.
The fundamental relation between the inputs and outputs can be obtained by combining Eq.~\eqref{app:multi-input} and~\eqref{app:multi-output}
\begin{equation}
    \sum_k \big[ \sqrt{\kappa_{ik}}\A{ik,\mathrm{out}}\pqty{t} + \sqrt{\kappa_{ik}}\A{ik,\mathrm{in}}\pqty{t} \big]= \kappa_i\A{i}\pqty{t} \ .
\end{equation}
In practice, stronger conditions
\begin{equation}
    \label{app:multi_input_output}
    \sqrt{\kappa_{ik}}\A{ik,\mathrm{out}}\pqty{t} + \sqrt{\kappa_{ik}}\A{ik,\mathrm{in}}\pqty{t} = \kappa_{ik}\A{i}\pqty{t}  , \quad \forall k
\end{equation}
are enforced.
\section{Equation of motion}
\label{app:eom}
In this section, we write down a compact matrix form of the quantum Langevin equation (re-derived in Appendix.~\ref{app:input_output}) in the displacement frame. Plugging the transducer Hamiltonian (Eq.~\eqref{app:Hlin_s}) into the quantum Langevin equation (Eq.~\eqref{app:multi-input}), we have:
\begin{subequations}
\label{app:heisenberg_eom}
\begin{align}
    \dot{a}_\mm &= -(i\omega_\mm +\frac{\kappa_\mm}{2})a_\mm - iG^*_\oo(t)a_\oo - i G_\oo(t)\Ad{\oo} - iG^*_\ee(t)a_\ee - i G_\ee(t)\Ad{\ee} - i|G_\oo(t)|^2 - i|G_\ee(t)|^2 + \sum_k \sqrt{\kappa_{\mathrm{m}k}}a_{\mm k,\mathrm{in}} \label{eq:eom_b}\\
    \dot{a}_\mm^\dagger &= (i\omega_\mm -\frac{\kappa_\mm}{2})a_\mm^\dagger + iG^*_\oo(t)a_\oo + i G_\oo(t)\Ad{\oo} + iG^*_\ee(t)a_\ee + i G_\ee(t)\Ad{\ee} + i|G_\oo(t)|^2 + i|G_\ee(t)|^2+\sum_k \sqrt{\kappa_{\mathrm{m}k}}a_{\mm k,\mathrm{in}}^\dagger \label{eq:eom_bd}\\
    \dot{a}_\oo &= -i\Delta_\oo a_\oo -\frac{\kappa_\oo}{2} a_\oo - iG_\oo(t)\pqty{a_{\mm}+a^\dagger_{\mm}} + \sum_k\sqrt{\kappa_{\mathrm{o}k}}\Ain{\oo k}\label{eq:eom_ao}\\
    \dot{a}^\dagger_\oo &=i\Delta_\oo \Ad{\oo} -\frac{\kappa_\oo}{2} \Ad{\oo}+ iG^*_\oo(t)\pqty{a_{\mm}+a^\dagger_{\mm}}+\sum_k\sqrt{\kappa_{\mathrm{o}k}}\Adin{\oo k}\label{eq:eom_aod}\\
    \dot{a}_\ee &=-i\Delta_\ee a_\ee -\frac{\kappa_\ee}{2} a_\ee - iG_\ee(t)\pqty{a_{\mm}+a^\dagger_{\mm}}+\sum_k \sqrt{\kappa_{\mathrm{e}k}}\Ain{\ee k}\label{eq:eom_ae}\\
    \dot{a}^\dagger_\ee & = i\Delta_\ee \Ad{\ee} -\frac{\kappa_\ee}{2} \Ad{\ee}+ iG^*_\ee(t)\pqty{a_{\mm}+a^\dagger_{\mm}}+\sum_k\sqrt{\kappa_{\mathrm{e}k}}\Adin{\ee k} \label{eq:eom_aed}\ ,
\end{align}
\end{subequations}
where we omit the hat symbol for the operators. Eqs.~\eqref{app:heisenberg_eom} can be rewritten into a matrix form by defining the \emph{state vector} as
\begin{equation}
    \label{app:state_space}
    \mathbf{a}(t)=\bqty{\A{\oo}\pqty{t}, \A{\ee}\pqty{t}, \Ad{\oo}\pqty{t}, \Ad{\ee}\pqty{t}, \A{\mm}\pqty{t}, \Ad{\mm}\pqty{t}}^\mathrm{T} \ ,
\end{equation}
and the \emph{input} and \emph{output vectors} as
\begin{subequations}
\label{app:input_output_vec}
\begin{align}
    \mathbf{a}_{\iin}\pqty{t} &= \Big[\hat{a}_{\oo1,\iin}\pqty{t},\ldots,\hat{a}_{\oo k_{\oo},\iin}\pqty{t}, \hat{a}_{\ee1,\iin}\pqty{t},\ldots,\hat{a}_{\ee k_{\ee},\iin}\pqty{t},\hat{a}^\dagger_{\oo1,\iin}\pqty{t},\ldots, \hat{a}^\dagger_{\ee1,\iin}\pqty{t},\ldots \notag\\
    &\quad\quad \hat{a}_{\mm 1, \iin}\pqty{t},\ldots,\hat{a}_{\mm k_{\mm}, \iin}\pqty{t},\hat{a}^\dagger_{\mathrm{\mm 1, in}}\pqty{t},\ldots,\hat{a}^\dagger_{\mm k_{\mm},\mathrm{in}}\pqty{t} \Big]^\mathrm{T} \\
    \mathbf{a}_{\out}\pqty{t} &= \Big[\A{\oo 1,\out}\pqty{t},\ldots,\A{\oo k_\oo, \out}, \A{\ee 1,\out}\pqty{t},\ldots,\Ad{\ee k_\ee, \out}\pqty{t},\Ad{\oo 1,\out}\pqty{t},\ldots , \Ad{\ee 1,\iin}\pqty{t},\ldots \notag\\
    &\quad\quad\A{\mm 1, \iin}\pqty{t},\ldots,\A{\mm k_{\mm},\out},\Ad{\mm 1,\out}\pqty{t}\ldots,\Ad{\mm k_{\mm},\out}\pqty{t}\Big]^\mathrm{T}   
    \ .    
\end{align}
\end{subequations}
The compact matrix form of Eqs.~\eqref{app:heisenberg_eom} is
\begin{equation}
    \label{app:eom_matrix}
    \mathbf{\dot{a}}(t) = \Ab(t)\mathbf{a}(t) + \Bb\mathbf{a}_{\mathrm{in}}\pqty{t} + \mathbf{v}\pqty{t} \ ,
\end{equation}
where
\begin{subequations}
\begin{align}
\Ab(t) &=  
    \begin{pmatrix}
        -i\Delta_\oo - \frac{\kappa_\oo}{2} & 0 & 0 & 0&-iG_\oo\pqty{t}&-iG_\oo\pqty{t}  \\
        0 & -i \Delta_\ee - \frac{\kappa_\ee}{2} & 0 & 0&-i G_\ee\pqty{t}&-i G_\ee\pqty{t}  \\
        0&0&i\Delta_\oo-\frac{\kappa_\oo}{2}&0&iG^*_\oo\pqty{t}&iG^*_\oo\pqty{t}\\
        0&0&0&i\Delta_e-\frac{\kappa_\ee}{2}&iG^*_\ee\pqty{t}&iG^*_\ee\pqty{t}  \\
        -iG^*_\oo\pqty{t} & -i G^*_\ee\pqty{t} & -iG_\oo\pqty{t} &-iG_\ee\pqty{t}&-i\omega_\mm-\frac{\kappa_\mm}{2}& 0  \\
        iG^*_\oo\pqty{t} & i G^*_\ee\pqty{t} & iG_\oo\pqty{t} &iG_\ee\pqty{t}&0& i\omega_\mm-\frac{\kappa_\mm}{2} 
    \end{pmatrix} \label{app:At}\\
    \Bb &=\begin{pmatrix}
        \mathbf{D}& 0 & 0 \\
        0 & \mathbf{D} & 0\\
        0 & 0 & \mathbf{M}\\
    \end{pmatrix} \\
    \mathbf{D} &=\begin{pmatrix}
        \sqrt{\kappa_{\oo 1}} &\cdots& \sqrt{\kappa_{\oo k_{\oo}}} & 0 &\cdots & 0  \\
        0 & \cdots & 0 & \sqrt{\kappa_{\ee 1}} &\cdots& \sqrt{\kappa_{\ee k_{\ee}} }
    \end{pmatrix}\\
        \mathbf{M} &=\begin{pmatrix}
        \sqrt{\kappa_{\mm 1}} &\cdots& \sqrt{\kappa_{\mm k_{\mm}}} & 0 &\cdots & 0  \\
        0 & \cdots & 0 & \sqrt{\kappa_{\mm 1}} &\cdots& \sqrt{\kappa_{\mm k_{\mm}} }
    \end{pmatrix}\\
    \mathbf{v}\pqty{t} &= \bigg[0,
    0,0,
    0,-ig_\oo\abs{\alpha_\oo\pqty{t}}^2 - ig_\oo\abs{\alpha_\ee\pqty{t}}^2,ig_\oo\abs{\alpha_\oo\pqty{t}}^2 + ig_\ee\abs{\alpha_\ee\pqty{t}}^2
    \bigg]^\mathrm{T}\ .\label{app:v}
\end{align}
\end{subequations}
Additionally the input-output relation ( Eq.~\eqref{app:multi_input_output}) has a matrix form given by
\begin{equation}
 \label{app:input_output_matrix}
 \mathbf{a}_{\mathrm{out}}\pqty{t} + \mathbf{a}_{\mathrm{in}}\pqty{t} = \mathbf{B}^\mathrm{T}\mathbf{a}\pqty{t} \ .
\end{equation}

For convenience, we will denote the first input and output of each cavity as $\A{i,\iin}\pqty{t}$ and $\A{i,\out}\pqty{t}$. Furthermore, we assume $\A{\oo,\iin}\pqty{\w}$ and $\A{\oo,\out}\pqty{\w}$ ( $\A{\ee,\iin}\pqty{\w}$ and $\A{\ee,\out}\pqty{\w}$) represent the incident and outgoing signal photons on the optical (electrical) cavity and call them the \emph{external inputs} and \emph{outputs}, respectively. The corresponding coupling strengths are called \emph{external coupling strengths} and denoted by $\kappa_{i,\mathrm{ex}}$. The rest of inputs/outputs are assumed to be the incident noise/internal loss which we don't have control over. Unless otherwise specified, we will assume all the noise inputs to be vacuum fields.

\section{Transfer matrix}\label{app:freq_domain}
In this section, we calculate the transfer matrix of the optomechanical quantum transducer. Without loss of generality, we will assume the task is to transduce one photon from the external input of the electrical cavity to the external output of the optical cavity.

\subsection{Constant driving}
\label{app:freq_const}
If the control signals $\Omega_i\pqty{t}$ in Eq.~\eqref{app:intra_cavity} are time-independent, the steady state solutions of $\alpha_i\pqty{t}$ are also constants
\begin{equation}
    \label{app:steady_const}
    \alpha^{\mathrm{s}}_i\pqty{t} = \frac{2\Omega_i}{i\kappa_i-2\Delta_i} \ .
\end{equation}
Then it is most convenient to solve the equation of motion Eq.~\eqref{app:eom_matrix} in the frequency domain assuming $\Ab\pqty{t}$ takes its steady state value. Using the convention of $\int f\pqty{t} e^{i\omega t}\dd{t}$ for the Fourier transform, we write Eq.~\eqref{app:eom_matrix} in the frequency domain
\begin{equation}
    \label{app:eom_f_const}
    -i\omega \mathbf{a}\pqty{\omega} = \Ab \mathbf{a}\pqty{\omega} + \Bb \mathbf{a}_{\mathrm{in}}\pqty{\omega} \ .
\end{equation}
$\mathbf{v}\pqty{t}$ is ignored because its steady state contributes only to the zero frequency.
Combing Eq.~\eqref{app:input_output_matrix} and~\eqref{app:eom_f_const}, we obtain the transfer matrix $\mathbf{T}\pqty{\omega}$
\begin{equation}
    \mathbf{a}_{\mathrm{out}}\pqty{\omega} = \bqty{\Bb^{\mathrm{T}}\pqty{-i\omega\II - A}^{-1} \Bb- \II}\mathbf{a}_{\mathrm{in}}\pqty{\omega} \equiv \TT\pqty{\omega} \mathbf{a}_{\mathrm{in}}\pqty{\omega} \ .
\end{equation}
We denote the element of the transfer matrix $\mathbf{T}\pqty{\omega}$ using the scattering notation
\begin{equation}
    \A{\oo,\out}\pqty{\omega} = \TT_{\oo, \ee}\pqty{\omega} \A{\ee,\iin}\pqty{\omega} \ ,  \quad\A{\oo,\out}\pqty{\omega} = \TT_{\oo, \ee'}\pqty{\omega} \Ad{\ee,\iin}\pqty{\omega} \ ,\quad\Ad{\oo,\out}\pqty{\omega} = \TT_{\oo', \ee}\pqty{\omega} \A{\ee,\iin}\pqty{\omega} \ ,
\end{equation}
where the prime symbol in the subscript means the corresponding operator is an annihilation operator.
In a typical experimental setup, we would have $\Delta_\oo=\Delta_\ee=\omega_\mm$ and the transmission coefficient at $\omega_\mm$ is
\begin{equation}
    \TT_{\oo, \ee}\pqty{\omega_\mm}= \frac{32\omega_\mm\sqrt{\kappa_{\oo,\mathrm{ex}}}\sqrt{\kappa_{\ee,\mathrm{ex}}}\pqty{\kappa_\ee-4i\omega_\mm}\pqty{\kappa_\oo-4i\omega_\mm}G^*_\ee G_\oo}{64\omega^2_\mm\kappa_\oo\pqty{i\kappa_\oo+4\omega_\mm}\abs{G_\ee}^2+\kappa_\ee\pqty{i\kappa_\ee+4\omega_\mm}\pqty{-\kappa_\mm\kappa_\oo\pqty{\kappa_\mm-4i\omega_\mm}\pqty{\kappa_\oo-4i\omega_\mm}+64\omega^2_\mm\abs{G_\oo}^2}} \ ,
\end{equation}
where $G_\oo$ and $G_\ee$ all take their steady state values
\begin{equation}
    G_i = g_i\alpha^\mathrm{s}_i = \frac{2g_i\Omega_i}{i\kappa_i-2\w_\mm} \ .
\end{equation}

If the cavities are ideal and symmetric, i.e., $\kappa_{\ee,\mathrm{ex}}=\kappa_{\oo,\mathrm{ex}}=\kappa_\ee = \kappa_\oo=\kappa$, $g_\ee=g_\oo$, and $\Omega_\oo=\Omega_\ee$, the transmission coefficient and transduction efficiency becomes
\begin{equation}
    \label{app:const_Toe}
    \TT_{\oo, \ee}\pqty{\omega_\mm}=-1-\frac{ i\kappa}{4\w_\mm} \ , \quad \eta\pqty{\w_\mm} \equiv \abs{\TT_{\oo, \ee}\pqty{\omega_\mm}} = \sqrt{1+\frac{\kappa^2}{16\w^2_\mm}} \ .
\end{equation}
Other elements of the transfer matrix connecting $\A{\oo,\out}\pqty{\w_\mm}$ to the input operators are
\begin{equation}
    \label{app:const_T_other}
    \TT_{\oo,\oo}\pqty{\w_\mm}=-\frac{ i\kappa}{4\w_\mm},\quad \TT_{\oo,\oo'}\pqty{\w_\mm} = \TT_{\oo,\ee'}\pqty{\w_\mm} = -\frac{ i\kappa G}{4\w_\mm G^*} \ .
\end{equation}
From Eq.~\eqref{app:const_Toe} and~\eqref{app:const_T_other}, we can see the noiseless transduction can only be achieved in the resolved sideband limit $4\w_\mm \gg \kappa$.
\subsection{Oscillating drive}
\label{app:freq_osc}
If the control signals in Eq.~\eqref{app:intra_cavity} are oscillating functions $\Omega_i\pqty{t} = \Omega_i e^{-2i\omega_\mm t}$, the steady state solutions of $\alpha_i\pqty{t}$ are also single frequency oscillating functions
\begin{equation}
    \label{app:steady_alpha}
    \alpha^\mathrm{s}_i\pqty{t} = \frac{2\Omega_i}{4\omega_\mm-2\Delta_i+i\kappa_i}e^{-2i\omega_\mm t}\equiv \Omega^\mathrm{s}_i e^{-2i\omega_\mm t} \ .
\end{equation}
As a result, the steady state $A(t)$ matrix (Eq.~\eqref{app:At}) can be written as a Fourier series
\begin{equation}
    \Ab\pqty{t} = \AD + \Am e^{-2i\omega_\mm t} + \Ap e^{2i\omega_\mm t} \ ,
\end{equation}
where
\begin{subequations}
\begin{align}
    \AD &=  
    \begin{pmatrix}
        -i\Delta_\oo - \frac{\kappa_\oo}{2} & 0 & 0 & 0&0&0  \\
        0 & -i \Delta_\ee - \frac{\kappa_\ee}{2} & 0 & 0&0&0  \\
        0&0&i\Delta_\oo-\frac{\kappa_\oo}{2}&0&0&0 \\
        0&0&0&i\Delta_e-\frac{\kappa_\ee}{2}&0&0  \\
        0 & 0 & 0 &0&-i\omega_\mm-\frac{\kappa_\mm}{2}& 0  \\
        0 & 0 & 0 &0&0& i\omega_\mm-\frac{\kappa_\mm}{2} 
    \end{pmatrix} \\
    \Am &=  
    \begin{pmatrix}
        0 & 0 & 0 & 0&-iG_\oo&-iG_\oo  \\
        0 & 0 & 0 & 0&-i G_\ee&-i G_\ee  \\
        0&0&0&0&0&0\\
        0&0&0&0&0&0  \\
        0 & 0 & -iG_\oo & -iG_\ee &0& 0  \\
        0 & 0 & iG_\oo & iG_\ee &0& 0
    \end{pmatrix} \label{app:Am}\\
    \Ap &=  
    \begin{pmatrix}
        0 & 0 & 0 & 0&0&0  \\
        0 & 0 & 0 & 0&0&0  \\
        0&0&0&0&iG^*_\oo&iG^*_\oo\\
        0&0&0&0&iG^*_\ee&iG^*_\ee \\
        -iG^*_\oo & -i G^*_\ee & 0 &0&0& 0  \\
        iG^*_\oo & i G^*_\ee & 0 &0&0& 0
    \end{pmatrix} \ , \label{app:Ap}\\
\end{align}
\end{subequations}
and $G_\oo$, $G_\ee$ take their corresponding steady state amplitudes
\begin{equation}
    G_i = g_i \Omega^\mathrm{s}_i = \frac{2g_i\Omega_i}{4\omega_\mm-2\Delta_i+i\kappa_i} \ .
\end{equation}

The frequency domain equation of motion can then be written as 
\begin{equation}
    \label{app:shifted_eoq}
    -i\omega \mathbf{a}\pqty{\omega} = \AD\mathbf{a}\pqty{\omega}+\Am\mathbf{a}\pqty{\omega-2\omega_\mm} +\Ap\mathbf{a}\pqty{\omega+2\omega_\mm} + \Bb\mathbf{a}_{\mathrm{in}}\pqty{\omega} \ ,
\end{equation}
where again we ignore $\mathbf{v}\pqty{t}$. To obtain a solvable system, we plug the frequencies $\omega+2k\omega_\mm$ into Eq.~\eqref{app:shifted_eoq} and assume $\mathbf{a}_{\mathrm{in}}\pqty{\omega}$ is zero at all the sidebands
\begin{subequations}
\label{app:eom_expaned}
\begin{align}
    & \quad\vdots\notag \\
    -i\pqty{\omega-2N\omega_\mm} \mathbf{a}\pqty{\omega-2N\omega_\mm} &= \AD\mathbf{a}\pqty{\omega-2N\omega_\mm}+\Am\mathbf{a}\pqty{\omega-2\pqty{N+1}\omega_\mm} +\Ap\mathbf{a}\pqty{\omega-2\pqty{N-1}\omega_\mm} \label{app:upper_boundary}\\
    &\quad \vdots \notag\\
    -i\omega \mathbf{a}\pqty{\omega} &= \AD\mathbf{a}\pqty{\omega}+\Am\mathbf{a}\pqty{\omega-2\omega_\mm} +\Ap\mathbf{a}\pqty{\omega+2\omega_\mm} + B\mathbf{a}_{\mathrm{in}}\pqty{\omega} \label{app:central_band}\\
    &\quad \vdots\notag\\
    -i\pqty{\omega+2N\omega_\mm} \mathbf{a}\pqty{\omega+2N\omega_\mm} &= \AD\mathbf{a}\pqty{\omega+2N\omega_\mm}+\Am\mathbf{a}\pqty{\omega+2\pqty{N-1}\omega_\mm} +\Ap\mathbf{a}\pqty{\omega+2\pqty{N+1}\omega_\mm} \label{app:lower_boundary}\\
    &\quad \vdots\notag
\end{align}
\end{subequations}
Then we truncate Eq.~\eqref{app:eom_expaned} by assuming $
    \mathbf{a}\pqty{\omega+2k\omega_\mm} = 0
$ for all $k<-N$ and $k>N$. The resulting equations can be solved iteratively. For example, starting from the boundary equation~\eqref{app:upper_boundary}, we can solve for $\mathbf{a}\pqty{\omega-2N\omega_\mm}$ in terms of $\mathbf{a}\pqty{\omega-2\pqty{N-1}\omega_\mm}$, resulting in
\begin{equation}
    \mathbf{a}\pqty{\omega-2N\omega_\mm} = \XX{N}{-}\Ap\mathbf{a}\pqty{\omega-2\pqty{N-1}\omega_\mm} \ ,
\end{equation}
where $\XX{N}{-} = \bqty{-i\pqty{\omega-2N\omega_\mm}\II-\AD}^{-1}$.
The solution can be plugged into the equation of motion at frequency $\omega - 2(N-1)\omega_\mm$ to obtain
\begin{equation}
    \mathbf{a}\pqty{\omega-2\pqty{N-1}\omega_\mm} =\XX{N-1}{-} \Ap\mathbf{a}\pqty{\omega-2\pqty{N-2}\omega_\mm} \ ,
\end{equation}
where
\begin{equation}
    \XX{N-1}{-}=\bqty{ -i\pqty{\omega-2\pqty{N-1}\omega_\mm}\II-\AD-\Xm^{\bqty{N}}}^{-1} \ ,\quad
    \Xm^{\bqty{N}} = \Am \XX{N}{-} \Ap \ .
\end{equation}
By applying the above procedure iteratively, we end up with an recursive relation
\begin{equation}
    \mathbf{a}\pqty{\omega-2k\omega_\mm} = \XX{k}{-}\Ap\mathbf{a}\pqty{\omega-2\pqty{k-1}\omega_\mm} \ ,  
\end{equation}
where
\begin{equation}
    \XX{k}{-}=\bqty{-i\pqty{\omega-2k\omega_\mm}\II-\AD-\Xm^{\bqty{k+1}}}^{-1}, \quad \Xm^{\bqty{k+1}}=\Am \XX{k+1}{-} \Ap \ .
\end{equation}
A similar recursive relation can be obtained by starting from the other boundary Eq.~\eqref{app:lower_boundary}
\begin{equation}
    \mathbf{a}\pqty{\omega+2k\omega_\mm}=\XX{k}{+}\Am\mathbf{a}\pqty{\omega+2\pqty{k-1}\omega_\mm} \ ,
\end{equation}
where
\begin{equation}
    \XX{k}{+}=\bqty{-i\pqty{\omega+2k\omega_\mm}\II-\AD-\Xp^{\bqty{k+1}}}^{-1} ,  \quad \Xp^{\bqty{k+1}} = \Ap\XX{k+1}{+}\Am \ .
\end{equation}
Setting $k=1$, we have
$
    \mathbf{a}\pqty{\omega\mp2\omega_\mm}=\XX{1}{\mp}\mathbf{A}_{\pm}\mathbf{a}\pqty{\omega}
$. The central band Eq.~\eqref{app:central_band} can then be rewritten as 
\begin{equation}
    -i\omega \mathbf{a}\pqty{\omega} = \Bqty{\AD+\XI{1}{-}\pqty{\omega}+\XI{1}{+}\pqty{\omega}}\mathbf{a}\pqty{\omega}+ \Bb\mathbf{a}_{\mathrm{in}}\pqty{\omega} \ ,
\end{equation}
from which a transfer function can be obtained
\begin{equation}
    \label{app:osc_transfer}
    \mathbf{a}_{\mathrm{out}}\pqty{\w} = \bqty{\Bb^\mathrm{T} \pqty{-i\omega\II - \AD-\XI{1}{+}\pqty{\w}-\XI{1}{-}\pqty{\w}}^{-1} \Bb- \II}\mathbf{a}_{\mathrm{in}}\pqty{\w}\equiv \TT\pqty{\omega} \mathbf{a}_{\mathrm{in}}\pqty{\w} \ .
\end{equation}
For simplicity, we write the transfer matrix as $\TT\pqty{\omega}= \Bb^{\mathrm{T}} \mathbf{X} \Bb- \II$ by defining the shorthand notation $\mathbf{X}=\pqty{-i\omega\II - \AD-\XI{1}{-}\pqty{\w}-\XI{1}{+}\pqty{\w}}^{-1}$.

The transfer matrix from inputs at sideband frequencies $\mathbf{a}_{
\iin}\pqty{\w+2k\w_\mm}$ can also be iteratively calculated using the same procedure
\begin{equation}
    \label{app:sidebands_input_coupling}
    \mathbf{a}_{\mathrm{out}}\pqty{\w} = \TT\pqty{\omega} \mathbf{a}_{\mathrm{in}}\pqty{\w} + \sum_{k=1}^N \bqty{\TT_+^{[k]}\pqty{\w} \mathbf{a}_{\iin}\pqty{\omega + 2k\w_\mm} + \TT_-^{[k]}\pqty{\w}\mathbf{a}_{\iin}\pqty{\omega - 2k\w_\mm}} \ ,
\end{equation}
where
\begin{equation}
    \label{app:Tk}
    \TT_{\pm}^{\bqty{k}}= \Bb^\mathrm{T}\mathbf{X}\prod_{i=1}^{k} \pqty{\Ab_\pm\XX{i}{\pm}}\mathbf{B} \ .
\end{equation}
We prove in Sec.~\ref{sec:proof} of the main text that $\TT_\pm^{\bqty{k}}=0$ for $k>2$.

    \label{app:general_T_osc}

Again, in the limit of $\Delta_\oo=\Delta_\ee=\omega_\mm$, the transfer function of $N=1$ is given by
\begin{equation}
    \TT_{\oo,\ee}\pqty{\omega_\mm}=\frac{32i\omega_\mm\sqrt{\kappa_{\oo,\mathrm{ex}}}\sqrt{\kappa_{\ee,\mathrm{ex}}}G^*_\ee G_\oo}{-16i\omega_\mm\kappa_\oo\abs{G_\ee}^2+\kappa_\ee\pqty{\kappa_\mm\kappa_\oo\pqty{\kappa_\mm+4i\omega_\mm}-16i\omega_\mm\abs{G_\oo}^2}} \ .
\end{equation}
If the cavities are symmetric, i.e., $\kappa_{\ee,\mathrm{ex}}=\kappa_{\oo,\mathrm{ex}}=\kappa_\ee = \kappa_\oo=\kappa$, $g_\ee=g_\oo$, and $\Omega_\oo=\Omega_\ee$, the transmission coefficient and transduction efficiency become
\begin{equation}
    \label{app:N1}
    \TT_{\oo, \ee}\pqty{\omega_\mm}=-1 ,\quad \eta\pqty{\w_\mm} \equiv \abs{\TT_{\oo, \ee}\pqty{\omega_\mm}} = 1 \ .
\end{equation}
Eq.~\eqref{app:N1} corresponds to the perfect transduction. The expressions for other quantities can be obtained using~\cite{notebook}. 
\section{Added noise}
\label{app:added_noise}
In addition to the transduction efficiency, the added noise is another important metric for the bosonic transducer~\cite{Zeuthen2020-sr}. In this section we derive an explicit expression and a lower bound for the added noise. The total noise coming out from the optical cavity is defined as~\cite{Lau2020-si,Caves1982-qx}
\begin{equation}
    \label{app:add_noise_def}
    2\pi\eta^2\pqty{\omega}S_\mathrm{t}\pqty{\omega}\delta\pqty{\w- \w'}=\frac{1}{2}\expval{\acomm{\A{\oo,\out}\pqty{\w'}^\dagger}{\A{\oo,\out}\pqty{\w}}}-\Big\langle\A{\oo,\out}\pqty{\omega}\Big\rangle\Big\langle\A{\oo,\out}\pqty{\omega'}^\dagger\Big\rangle \ ,
\end{equation}
where the average $\expval{\cdot}$ is over the initial state of all the inputs and outputs.
The $2\pi$ factor on the LHS comes from the normalization factor of frequency domain commutation relation (Eq.~\eqref{app:input_normalization}). Eq.~\eqref{app:add_noise_def} can be simplified by replacing $\A{\oo,\out}\pqty{\w}$ with the scattering form
\begin{equation}
    \label{app:aout}
     \A{\oo,\out}\pqty{\w} = \TT_{\oo,\ee}\pqty{\w}\A{\ee,\iin}\pqty{\w} + \hat{J}\pqty{\w} \equiv  \TT_{\oo,\ee}\pqty{\w}\A{\ee,\iin}\pqty{\w} + \hat{J}_\mathrm{ex}\pqty{\w} + \hat{J}_\mathrm{im}\pqty{\w}
\end{equation}
where
\begin{align}
    \hat{J}_\mathrm{ex}\pqty{\w}=\TT_{\oo,\ee'}\pqty{\w}\Ad{\ee,\iin}\pqty{\w}+\TT_{\oo,\oo}\pqty{\w}\A{\oo,\iin}\pqty{\w}+\TT_{\oo,\oo'}\pqty{\w}\Ad{\oo,\iin}\pqty{\w}
\end{align}
and
\begin{equation}
    \label{app:implicit_noise}
    \hat{J}_\mathrm{im}\pqty{\w}=\sum_i \pqty{\TT_{\oo, i} \A{i,\iin} + \TT_{\oo, i'} \Ad{i,\iin} }\ .
\end{equation}
In the above equations, $\hat{J}\pqty{\w} \equiv \hat{J}_\mathrm{ex}\pqty{\w} + \hat{J}_\mathrm{im}\pqty{\w}$ contains all the unwanted modes being mixed into the output, among which $\hat{J}_\mathrm{ex}\pqty{\w}$ contains all external inputs in Eq.~\eqref{app:input_output_vec} (see the last paragraph of Sec.~\ref{app:eom}) and $\hat{J}_\mathrm{im}$ has all the other terms, i.e., the incident noises and the sideband couplings (Eq.~\eqref{app:sidebands_input_coupling}), meaning the index $i$ now includes the frequency shifts and goes through all the inputs on the right-hand side of Eq.~\eqref{app:sidebands_input_coupling} except for $\A{\ee,\iin}\pqty{\w}$ and $\Ad{\ee,\iin}\pqty{\w}$. The subscripts ``ex" and ``im" stand for ``explicit" and ``implicit" since the terms in $\hat{J}_\mathrm{im}$ all satisfy $\expval{\A{i,\iin}}=\expval{\Ad{i,\iin}}=0$ and are usually omitted when presenting the equation of motion Eq.~\eqref{app:heisenberg_eom}. Before proceeding, it is worth mentioning that, for operators which are zero, e.g. $\A{\mm, \iin}\pqty{\w}$ where $\w_\mm < 0$, we can set the corresponding transfer matrix element to zero while keeping the form of Eq.~\eqref{app:implicit_noise}. To avoid any confusion, we emphasize that the operators $\hat{a}\pqty{\w}^\dagger$ and $\hat{a}^\dagger\pqty{\w}$ are different ( see Eq.~\eqref{app:dagger_defination} for definitions).
The simplified RHS of Eq.~\eqref{app:add_noise_def} is given by
\begin{subequations}
\begin{align}
    \mathrm{RHS}&=\frac{1}{2}\expval{\acomm{\hat{J}\pqty{\w}}{\hat{J}\pqty{\w'}^\dagger}} + \TT_{\oo,\ee}\pqty{\w}\TT^*_{\oo,\ee}\pqty{\w'}\bqty{\expval{\acomm{\A{\ee,\iin}\pqty{\w}}{\A{\ee,\iin}\pqty{\w'}^\dagger}}/2- \Big\langle\A{\ee,\iin}\pqty{\w}\Big\rangle\Big\langle\A{\ee,\iin}\pqty{\w'}^\dagger\Big\rangle} \label{app:add_noise_eq_1}\\
    &\quad+\frac{\TT_{\oo,\ee}\pqty{\w}\TT^*_{\oo,\ee'}\pqty{\w'}}{2}\expval{\acomm{\A{\ee,\iin}\pqty{\w}}{\A{\ee,\iin}\pqty{-\w'}}}+\frac{\TT^*_{\oo',\ee'}\pqty{\w'}\TT_{\oo,\ee'}\pqty{\w}}{2}\expval{\acomm{\Ad{\ee,\iin}\pqty{\w}}{\Ad{\ee,\iin}\pqty{-\w'}}} \label{app:add_noise_eq_2}\\
    &\quad -\TT_{\oo,\ee}\pqty{\w}\TT^*_{\oo,\ee'}\pqty{\w'}\expval{\A{\ee,\iin}\pqty{\w}}\expval{\A{\ee,\iin}\pqty{-\w'}} -\TT^*_{\oo',\ee'}\pqty{\w'}\TT_{\oo,\ee'}\pqty{\w}\expval{\Ad{\ee,\iin}\pqty{\w}}\expval{\Ad{\ee,\iin}\pqty{-\w'}} \label{app:add_noise_eq_3}\ ,
\end{align}
\end{subequations}
where we assume $\expval{\A{\oo,\iin}\pqty{\w}}=0$ and all the inputs are not correlated with each other. Because the average is over a coherent state, lines~\eqref{app:add_noise_eq_2} and~\eqref{app:add_noise_eq_3} add up to $0$. The second term in line~\eqref{app:add_noise_eq_1} equals $\pi\eta^2\pqty{\w}\delta\pqty{\w-\w'}$, and sets a lower bound to $S_\mathrm{t}\pqty{\w}\ge \frac{1}{2}$, which is the quantum limit of the transduction noise and represents the vacuum fluctuation. Finally, the first term in line~\eqref{app:add_noise_eq_1} gives the expression of added noise
\begin{equation}
    \label{app:noise_S}
    2\pi\eta^2\pqty{\w}S\pqty{\w}\delta\pqty{\w-\w'} = \frac{1}{2}\expval{\acomm{\hat{J}\pqty{\w}}{\hat{J}\pqty{\w'}^\dagger}} \ .
\end{equation}

The RHS of Eq.~\eqref{app:noise_S} can be explicitly written out as
\begin{subequations}
    \label{app:acomm_J}
\begin{align}
    \frac{1}{2}\expval{\acomm{\hat{J}\pqty{\w}}{\hat{J}\pqty{\w'}^\dagger}} &= \frac{1}{2} \abs{\TT_{\oo,\ee'}\pqty{\w}}^2\expval{\acomm{\A{\ee,\iin}\pqty{\w}^\dagger}{\A{\ee,\iin}\pqty{\w'}}}  \\
    &+ \frac{1}{2} \abs{\TT_{\oo,\oo}\pqty{\w}}^2\expval{\acomm{\A{\oo,\iin}\pqty{\w}}{\A{\oo,\iin}\pqty{\w'}^\dagger}} + \frac{1}{2} \abs{\TT_{\oo,\oo'}\pqty{\w}}^2\expval{\acomm{\A{\oo,\iin}\pqty{\w}^\dagger}{\A{\oo,\iin}\pqty{\w'}}} \label{app:reflection_error} \\
    &+ \frac{1}{2}\sum_i\Bqty{ \abs{\TT_{\oo,i}\pqty{\w}}^2\expval{\acomm{\A{\oo,\iin}\pqty{\w}}{\pqty{\A{\oo,\iin}\pqty{\w'}}^\dagger}} +  \abs{\TT_{\oo,i'}\pqty{\w}}^2\expval{\acomm{\pqty{\A{i,\iin}\pqty{\w}}^\dagger}{\A{i,\iin}\pqty{\w'}}}} \label{app:implicit
    _error}\ .
\end{align}
\end{subequations}
Because the input of the electrical cavity is a coherent state with a single photon and all other inputs are the vacuum state, the added noise becomes
\begin{equation}
    \label{app:added_noise_expression}
    \eta^2\pqty{\w}S\pqty{\w}=\frac{3}{2}\abs{\TT_{\oo,\ee'}\pqty{\w}}^2 + \frac{1}{2}\abs{\TT_{\oo,\oo}\pqty{\w}}^2+\frac{1}{2}\abs{\TT_{\oo,\oo'}\pqty{\w}}^2 +\frac{1}{2}\sum_i \pqty{\abs{\TT_{\oo,i}\pqty{\w}}^2+\abs{\TT_{\oo,i'}\pqty{\w}}^2 } \ . 
\end{equation}
A lower bound of $S\pqty{\w}$ can be derived by allowing the input at the optical cavity to be a squeezed vacuum state~\cite{Lau2020-si}
\begin{subequations}
    \label{app:added_noise_inequality}
\begin{align}
    \eta^2\pqty{\w}S\pqty{\w}&\ge\frac{3}{2}\abs{\TT_{\oo,\ee'}\pqty{\w}}^2 + \frac{1}{2}\abs{\abs{\TT_{\oo,\oo}\pqty{\w}}^2-\abs{\TT_{\oo,\oo'}\pqty{\w}}^2} + \frac{1}{2}\sum_i \abs{\abs{\TT_{\oo,i}\pqty{\w}}^2-\abs{\TT_{\oo,i'}\pqty{\w}}^2 } \label{app:s_line1}\\
    &\ge\frac{3}{2}\abs{\TT_{\oo,\ee'}\pqty{\w}}^2 + \frac{1}{2}\abs{1-\abs{\TT_{\oo,\ee}\pqty{\w}}^2+\abs{\TT_{\oo,\ee'}\pqty{\w}}^2} \label{app:s_line2}\ ,
\end{align}
\end{subequations}
where we make use of
\begin{equation}
    \sum_i\abs{\abs{\TT_{\oo, i}}^2 - \abs{\TT_{\oo, i'}}^2} \ge \abs{\sum_i \abs{\TT_{\oo,i}}^2 - \abs{\TT_{\oo,i'}}^2}
\end{equation}
and
\begin{equation}
    \label{app:preservation_comm}
    1=\abs{\TT_{\oo,\ee}\pqty{\w}}^2-\abs{\TT_{\oo,\ee'}\pqty{\w}}^2 +\abs{\TT_{\oo,\oo}\pqty{\w}}^2-\abs{\TT_{\oo,\oo'}\pqty{\w}}^2 +\sum_i\abs{\TT_{\oo,i}\pqty{\w}}^2-\abs{\TT_{\oo,i'}\pqty{\w}}^2
\end{equation}
from line~\eqref{app:s_line1} to~\eqref{app:s_line2}. 
Eq.~\eqref{app:preservation_comm} comes from the preservation of commutation $\comm{\A{\oo,\out}\pqty{\w}}{\A{\oo,\out}\pqty{\w}^\dagger}=2\pi$.
A compact form of the lower bound is given by
\begin{equation}
\label{app:S_lower_bound}
S\pqty{\w} \ge \frac{3}{2}R^2\pqty{\w} + \abs{\frac{1-\eta^2\pqty{\w}}{2\eta^2\pqty{\w}}+\frac{R^2\pqty{\w}}{2}}
\end{equation}
with $R^2\pqty{\w}=\abs{\TT_{\oo,\ee'}\pqty{\w}}^2/\eta^2\pqty{\w}$. The equality can only be achieved by squeezing every input which has non-zero transmission coefficient to the output.

\section{Numerical values}
\label{app:num_vals}
Table~\ref{tab:num_vals} provides all the numerical values used in the paper, which are chosen according to a real experiment~\cite{higginbotham2018harnessing}. The coupling strength to the internal loss ports of the EM cavities are calculated by $\kappa_{i,\mathrm{int}}=\kappa_i-\kappa_{i,\mathrm{ex}}$. We always assume that the mechanical cavity has a single input and output, i.e.,  $\kappa_{\mm,\mathrm{ex}}=\kappa_\mm$.
\begin{table}[h]
    \centering
    \begin{tabular}{|c|c|c|c|c|c|}
    \hline
    parameter & value & parameter & value & parameter & value\\
    \hline
    $\Delta_\oo/2\pi$& 1.11 \ MHz & $\Delta_\ee/2\pi$& 1.47 \ MHz & $\omega_\mm/2\pi$& 1.4732 \ MHz\\
    $\kappa_\oo/2\pi$ & 2.1 \ MHz & $\kappa_\ee/2\pi$&2.5\ MHz & $\kappa_\mm/2\pi$& 11 \ Hz \\
    $\kappa_{\oo,\mathrm{ex}}/2\pi$ & 1.1 \ MHz & $\kappa_{\ee,\mathrm{ex}}/2\pi$& 2.3\ MHz& $\kappa_{\mm,\mathrm{ex}}/2\pi$& 11 Hz\\
    $g_\oo/2\pi$&6.6 \ Hz&$g_\ee/2\pi$& 3.8 \ Hz&   &  \\
    \hline
    \end{tabular}
    \caption{Numerical values of the parameters used in simulations.}
    \label{tab:num_vals}
\end{table}

The driven signal amplitudes $\Omega_i$ are also chosen according to Ref.~\cite{higginbotham2018harnessing}. In the reference, the reported experimental $\Gamma_\oo$ and $\Gamma_\ee$ values range from $2\pi \times  95 \mathrm{Hz}$ to $2\pi \times  725 \mathrm{Hz}$. Because $\Gamma_i = 4g_i^2 n_i / \kappa_i$, the steady state photon number $n_i$ inside each EM cavity needs to be at least on the order of $10^{6}$ to achieve the reported $\Gamma_i$. We know $n_i=\abs{\alpha_i}^2$, and for the constant driving case, $\alpha_i$ is given by $2\Omega_i/\pqty{i\kappa_i-2\Delta_i}$ (see Eq.~\eqref{app:steady_const}). So using the parameter values provided by table~\ref{tab:num_vals}, we estimate that a minimum driving strength on the order of $1000 \mathrm{MHz}$ is needed to achieve the reported $\Gamma_i$. Therefore, to make sure our choice of $\Omega_i$ is attainable, we chose most their numerical values below $1000 \mathrm{MHz}$.
\section{Results with more sidebands}
\label{app:more_sidebands}
In this section, we report the complementary figure to Fig. 1 and 2 in the main text. We numerically calculate the transduction efficiency and added noise of both the symmetric and realistic quantum transducers with the same parameters but more sidebands. We observe that the curves converge at $N=2$; namely, for more than four sidebands ($N > 2$) the resulting curves remain identical to the $N=2$ result. The numerical results confirm the proof presented in Sec.~\ref{sec:proof} of the main text. Our findings are shown in Fig.~\ref{fig:more_sidebands} and can be reproduced by the companion Mathematica notebook~\cite{notebook}.
\begin{figure}[ht]
     \centering
     \subfigure[]{
        \includegraphics[width=0.4\textwidth]{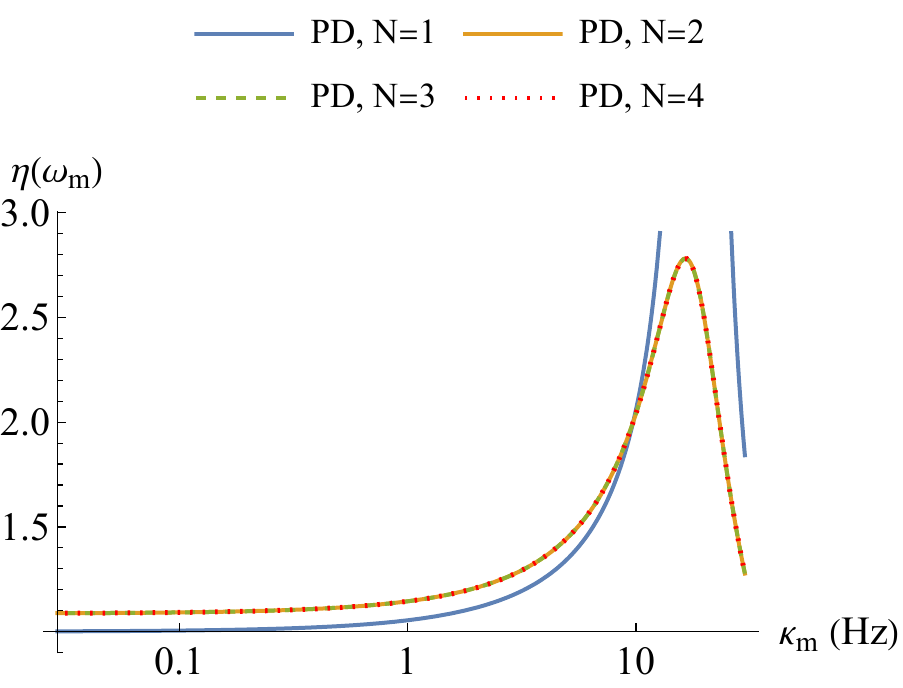}
        }
     \subfigure[]{
        \includegraphics[width=0.4\textwidth]{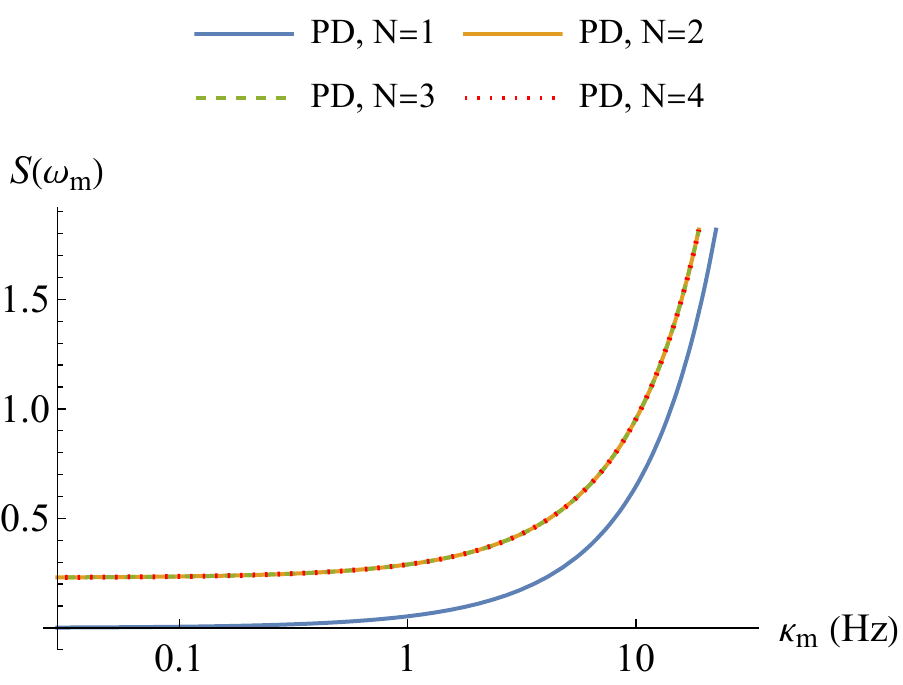}
        }
    \\
    \subfigure[]{\includegraphics[width=0.4\textwidth]{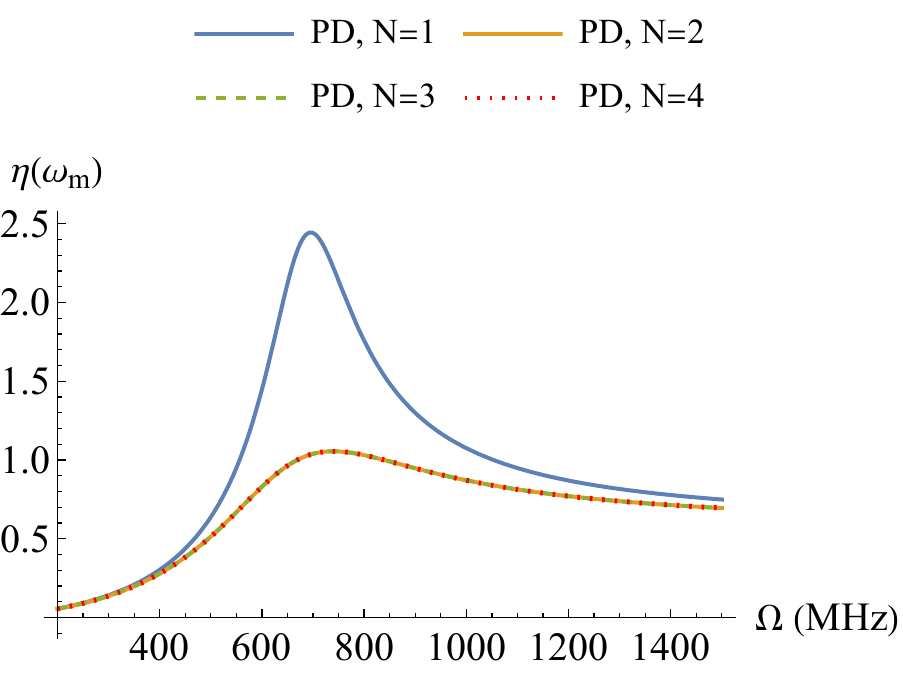}}
    \subfigure[]{\includegraphics[width=0.4\textwidth]{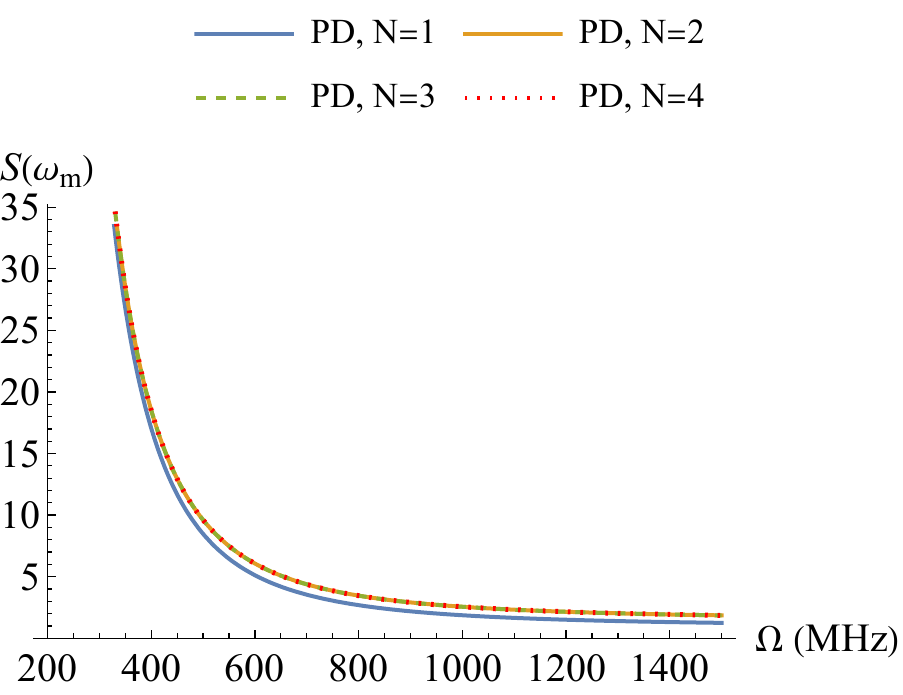}}
    \caption{Complementary figure to Fig. 1 and 2 in the main text. (a), (b): Transduction efficiency $\eta\pqty{\w_\mm}$ and added noise $S\pqty{\w_\mm}$ v.s. $\kappa_\mm$ of a symmetric transducer, complement Fig. 1 of the main text; (c), (d): Transduction efficiency $\eta\pqty{\w_\mm}$ and added noise $S\pqty{\w_\mm}$ v.s. driving amplitude $\Omega$ of a realistic optomechanical quantum transducer, , complement Fig. 2 of the main text.}
    \label{fig:more_sidebands}
\end{figure}
\FloatBarrier
\twocolumngrid
\bibliographystyle{apsrev4-2}
\bibliography{optomech.bib}

\begin{thebibliography}{34}
\bibitem [{Fou()}]{Fourier}%
  \BibitemOpen
  \href@noop {} {\bibinfo {title} {{We use $\int f(t) e^{i\omega t}
  \mathrm{d}t$ as the Fourier transform convention.}}}\BibitemShut {Stop}%
\end{thebibliography}
\end{document}